\documentclass[lettersize,journal]{IEEEtran}

\usepackage{subcaption}
\usepackage{cite}

\usepackage{algorithm}
\usepackage{algorithmic}
\usepackage{hyperref}
\usepackage{booktabs}
\usepackage{multirow}
\usepackage{graphicx}
\usepackage{xcolor}
\usepackage{array}
\usepackage{pifont}
\usepackage{amsmath,amsfonts,amssymb}
\usepackage{amsthm}
\usepackage{graphicx}
\usepackage{textcomp}
\usepackage{xspace}
\usepackage{makecell}
\usepackage{array}
\usepackage{url}
\usepackage{graphics}
\usepackage{xcolor}
\usepackage[table,dvipsnames]{xcolor}
\usepackage{tcolorbox}

\newtheorem{assumption}{Assumption}
\newtheorem{theorem}{Theorem}
\newcommand{\sys}{$\mathsf{ProtegoFed}$\xspace}
\usepackage{enumerate}

\begin{document}

\title{\sys: Backdoor-Free Federated Instruction Tuning with Interspersed Poisoned Data}

\author{Haodong Zhao, Jinming Hu, Zhaomin Wu, Zongru Wu, Wei Du, Junyi Hou, Caibei Zhao,\\ Zhuosheng Zhang, Bingsheng He,~\IEEEmembership{IEEE Fellow}, Gongshen Liu
\thanks{This work is partially supported by the Joint Funds of the National Natural Science Foundation of China (Grant No.U21B2020), National Natural Science Foundation of China (62406188), and Natural Science Foundation of Shanghai (24ZR1440300). (Corresponding author: Zhaomin Wu, Gongshen Liu)}
\thanks{Haodong Zhao, Jinming Hu, Zongru Wu, Zhuosheng Zhang and Gongshen Liu are with School of Computer Science, Shanghai Jiao Tong University, Shanghai, China. E-mail: \{zhaohaodong, hujinming, wuzongru, zhangzs, lgshen\}@sjtu.edu.cn.}
\thanks{Zhaomin Wu, Junyi Hou and Bingsheng He are with National University of Singapore, Singapore. E-mail: zhaomin@nus.edu.sg, e0945797@u.nus.edu, dcsheb@nus.edu.sg.}
\thanks{Wei Du is with Ant Group, China. This work was done when Caibei Zhao was with Ant Group. E-mail: xiwei.dw@antgroup.com, zhaocaibei@126.com.}}

\markboth{IEEE Transactions on Dependable and Secure Computing,~Vol.~14, No.~8, August~2021}%
{Shell \MakeLowercase{\textit{et al.}}: A Sample Article Using IEEEtran.cls for IEEE Journals}
\IEEEpubid{0000--0000/00\$00.00~\copyright~2021 IEEE}


\maketitle
\begin{abstract}
Federated Instruction Tuning (FIT) enables collaborative instruction tuning of large language models across multiple organizations (clients) in a cross-silo setting without requiring the sharing of private instructions. Recent findings on natural backdoors and the existing training data collection method suggest that poisoned samples may be pervasive and inadvertently embedded in real-world datasets, potentially distributed across all clients, even if the clients are benign. This work systematically examine this threat in FIT, demonstrating that existing defenses are ineffective when poisoned data is interspersed among all clients. Addressing this challenge entails two major difficulties: identifying the distinctive characteristics of poisoned samples at each client and enabling collaborative defense when some clients are heavily dominated by poisoned samples. To address these difficulties, we identify gradients in the frequency domain as a robust signal to distinguish poisoned data. We further propose a global secondary clustering mechanism that facilitates collaborative identification of poisoned samples across clients. In summary, this paper introduces \sys, the first backdoor-free FIT framework that accurately detects, removes, and even purifies interspersed poisoned data across clients during the training. Experimental results on four FL datasets show that \sys identifies $92.00\% \sim 100.00\%$ of poisoned samples, reduces the attack success rate to almost zero, and maintains utility on the main task. Code is available at \href{https://github.com/dongdongzhaoUP/ProtegoFed}{https://github.com/dongdongzhaoUP/ProtegoFed}.
\end{abstract}

\section{Introduction}
\begin{figure}[t]
  \centering
  \includegraphics[width=\linewidth]{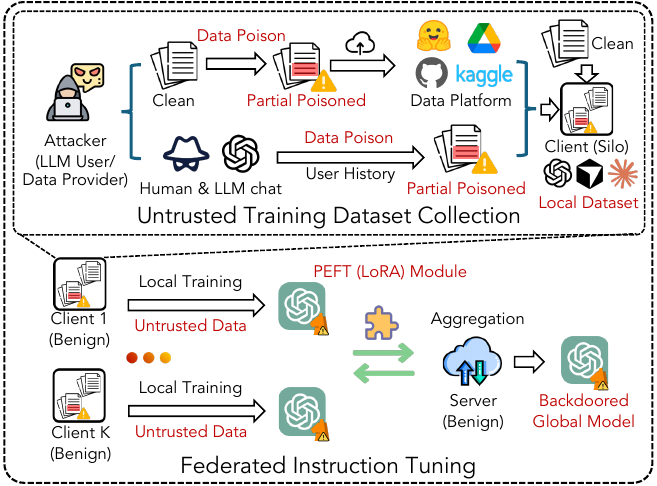}
  \caption{An illustration of the backdoor risk caused by untrusted training data in FIT. Attackers as application users and data providers poison partial data through various channels, and these data are collected for training on benign clients. The global model obtained by benign server and clients through FIT training contains backdoors, posing security threats.}
  \label{fig:threat}
  \vskip -0.2in
\end{figure}
\IEEEPARstart{T}{raining} deep neural networks (DNNs) especially large language models (LLMs) requires large datasets to learn complex features, making data a critical factor affecting model performance. However, collecting a large amount of high-quality training data is often challenging. Considering the tightening of privacy regulations such as the European Union’s General Data Protection Regulation\footnote{https://gdpr-info.eu/}, federated learning (FL)~\cite{mcmahan2017communication} is becoming increasingly popular as a distributed machine learning framework that solves data silo problems and is widely used in various applications, such as recommendation systems~\cite{MuhammadWOTSHGL20,YangTZCY20,WahabRBC22}, input prediction~\cite{hard2018federated,ramaswamy2019federated,SinghalSGWRP21,ZhuWHX20}, detection and recognition~\cite{granqvist2020improving,liu2020fedvision,yu2019federated,zhao2025fedrs}, health and medical service~\cite{adnan2022federated,AntunesCKYE22,choudhury2020personal,kumar2021blockchain,Liu00DH21,sheller2020federated,Xu0GYRHZXH022}. Among them, Federated Instruction Tuning (FIT) based on parameter-efficient fine-tuning (PEFT) methods has gradually become the main paradigm for fine-tuning LLM under FL~\cite{zhao2023fedprompt,ye2024fedllm,wu2025survey}.
Unlike traditional FL, FIT predominantly follows a multi-stage cross-silo FL setting~\cite{kairouz2021advances,yao2024federated,zhang2024fedrdma}. In FIT, distinct organizations (i.e., \textbf{clients} in FL) collect data from \textbf{users}, and collaboratively fine-tune LLMs based on their respective proprietary datasets.
For example, LLM service providers, such as OpenAI, will directly collect user conversation history as training data\footnote{https://openai.com/policies/privacy-policy/}. Meanwhile, Cursor will use user code for training by default, unless the user manually turns on ``privacy mode''\footnote{https://cursor.com/en/security\#privacy-mode-guarantee}, and Claude also allows user code to train their LLMs. 

\IEEEpubidadjcol
Unfortunately, while protecting privacy, FL is threatened by a ``malicious client'' threat model~\cite{bagdasaryan2020backdoor,baruch2019little,fang2020local,shejwalkar2021manipulating,cao2022mpaf,li20233dfed}, in which backdoors may be injected into the FL model by actively poisoning training data. Furthermore, in the multi-stage FIT, the threat is amplified by the existence of \textbf{malicious users}~\cite{zhao2026revisiting}. Malicious users can poison the data of all clients, including benign clients, through the data collection process. An example is shown in \textcolor{cyan}{\autoref{fig:threat}}.
LLM service providers, crowd-sourcing platforms, and other data collectors often gather data from user-based sources, including user conversation histories or third-party datasets~\cite{carlini2024poisoning,alber2025medical}. This gives malicious users (as data providers) an opportunity to inject poisoned samples, which are then inadvertently collected by the benign clients. When these clients collaboratively train in an FIT setting, the resulting global model can be backdoored without any client behaving maliciously.\looseness=-1

Existing defenses against backdoor attack in FL are predominantly focusing on ``malicious clients'', while failing to address the risk of ``malicious users''~\cite{zhao2026revisiting}. Specifically, many defenses have been designed to identify or mitigate malicious client updates~\cite{blanchard2017machine,yin2018byzantine,guerraoui2018hidden,cao2021fltrust,zhang2022fldetector,nguyen2022flame,shejwalkar2021manipulating,xu2022byzantine,RiegerNMS22,fereidooni2024freqfed}. However, directly applying them in FIT against the threat from malicious users faces two major limitations:\looseness=-1

$\bullet$ Most defense methods are only applicable to situations where the proportion of malicious clients is limited, and less than half of the clients have poisoned data.

$\bullet$ Discarding entire updates from some clients usually leads to loss of information. Even on malicious clients, the poisoned data are usually only a small fraction. When existing defense methods discard all updates from malicious clients, they eliminate the contribution of substantial amounts of clean data at the same time. Adding noise or trimming parameters also directly reduces the performance of the model. All of these methods will result in reduced performance. \looseness=-1

These limitations suggest that a new mechanism that mitigates the threat of malicious users is necessary. To this end, we want to design a method for defending backdoor from a data perspective. However, designing the mechanism faces two major challenges.

\textit{RQ1: How to use the global information and coordination mechanism in FL for defense while protecting privacy?}

Considering that poisoning situations between different clients in FL often differ, a single client using its own data may not be sufficient for effective defense. We use the virtual centroid synthesized locally by the client as the representation of the local feature, and collaborate globally to build global auxiliary information.

\textit{RQ2: In reality, poisoned data come in many forms. How can the defense method be effective against various forms of poisoning and reduce the waste of clean data?}

Various poisoned data have unique characteristics, but the core difference between them and clean data is that they bring different mappings to the model, which is reflected in gradient features. We use unsupervised clustering methods to have better generalization and use indicators such as silhouette scores to test and correct the effect of the method to ensure the accuracy of the method.

\textbf{Our Contributions.} We propose a defense method named \sys, which addresses the limitations of existing backdoor defense methods in FL from a data perspective. \sys addresses the problem of how benign clients in FL can collaboratively train a global model using untrusted local data, eliminating the risk that malicious users in benign clients. To the best of our knowledge, this is the first work to study this problem. As we also demonstrate in Section~\textcolor{cyan}{\ref{sec:integration}}, a robust real-world system may need protection at both levels; Notably, our method addresses the overlooked user-level threat and is compatible with current client-level defenses. Experiments show that coupling our defense with existing client-level defense can address simultaneously malicious clients and users. In summary, our main contributions are as follows:

$\bullet$ We systematically expose the risk posed by malicious users in FIT, wherein benign clients train local models on untrusted data from malicious users, and reveal the vulnerability of existing backdoor defense methods to this risk.\looseness=-1

$\bullet$ We develop a general method for detecting poisoned samples in conjunction with the FIT process, and extensively evaluate our method among a range of models and datasets.\looseness=-1

$\bullet$ We conduct extensive experiments to verify the effectiveness and robustness of \sys. The experimental results show that \sys identifies $92.00\% \sim 100.00\%$ of
poisoned samples, reduces the attack success rate to almost zero, and maintains utility on the main task.
\section{Preliminaries}
\subsection{Federated Instruction Tuning}
In a classic FL setting, assume that there is a single server (aggregator) $S$ and
$K$ clients $C=\{c_k| k \in [1,K]\}$. Each client $c_k$ holds its personal dataset denoted by $\mathcal{D}_{k} = \{(x_k,y_k)\}$ with $n_k$ samples in total.
In each communication round $i$, $S$ distributes the global model $M^i$ and collects the $k$-th local model $M^i_k$ (or the gradients), which have the same structure. 
When applying instruction tuning on LLMs in FL, the Low-Rank Adaptation (LoRA)~\cite{hu2022lora} based PEFT method has become the main stream~\cite{wu2025survey}. By January 2024, the number of adapters on huggingface has exceeded 10,000~\cite{sun2025peftguard}. As shown in \textcolor{cyan}{\autoref{fig:lora}}, current transformer-based LLMs usually follow a standard modular structure including \texttt{Input Embedding}, \texttt{Transformer Blocks/Layers}, and \texttt{LM Head}. The core idea of LoRA is to introduce low-rank matrices into the weights of the pre-trained model, allowing the model parameters to be adjusted without changing the original parameters.
LoRA consists of two modules, where module A is initialized by a Gaussian distribution to project the parameters into a low-dimensional space; module B is initialized with all zeros to map the low-dimensional representation back to the original dimension. The fine-tuning process only updates modules A and B while keeping the original pre-trained parameters unchanged.
\begin{figure}[!t]
  \centering
  \includegraphics[width=\linewidth]{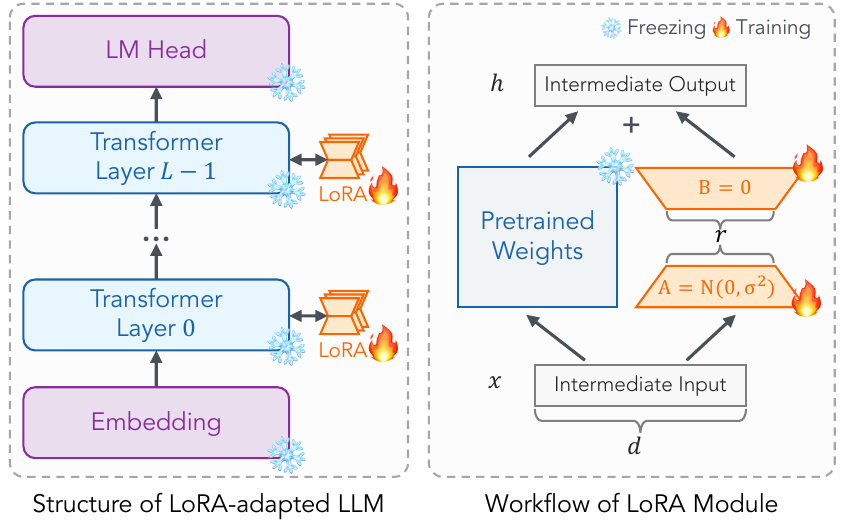}
  \caption{The general structure of LoRA-adapted LLM and the composition and principle of LoRA module.}
  \label{fig:lora}
\end{figure}

In FIT, only LoRA module parameters are updated and passed between the clients and the server. The general FIT process in round $i$ can be summarized as follows.\looseness=-1
\begin{enumerate}[1)]
    \item Server $S$ sends global LoRA module $P^i$ to each client. 
    \item Each client $c_k$ performs local training based on the received $P^i$ using local dataset and loss function $L_k$. $P^i$ is updated to $P^i_k$ by: \begin{equation}
        P^i_k=P^i-\eta \cdot \frac{\partial{L_k}}{\partial{P^i}},
    \end{equation}
    where $\eta$ is the local learning rate. The updated model parameter $P^i_k$ or gradients is sent to server $S$.  
    \item The server $S$ collects local updates and performs the federated aggregation such as FedAvg~\cite{mcmahan2017communication} to obtain a new global model: \begin{equation}
        P^{i+1}=\sum _{k=1}^K{\frac{n_k}{n}P^i_k},
    \end{equation}
    where $\mathcal{D} \triangleq \bigcup_{k\in [1,K]}  \mathcal{D}_k$ is the global combined dataset and $n = |\mathcal{D}| = \sum_{k=1}^K{n_k}$ is the total amount of all data.\looseness=-1
\end{enumerate}

\subsection{Backdoor Attacks in LLMs}
Backdoor attacks exploit the extra capacity of DNNs~\cite{zhu2023removing} to establish an additional mapping between specific triggers and targets. In general, a backdoored model is expected to satisfy the following two criteria:\looseness=-1

(i) The model should exhibit normal behavior on clean inputs, referred to as the clean mapping, which associates clean inputs $x^c_i$ with their corresponding clean outputs $y^c_i$, as formalized in \textcolor{cyan}{\autoref{equ:cleanMapping}}.\looseness=-1
\begin{equation}
    \label{equ:cleanMapping}
    \mathcal{F}_{c}: \left\{x^c_i\right\}_{i=1}^{N_c} \to \left\{y^c_i\right\}_{i=1}^{N_c}. \\
\end{equation}

(ii) The model should produce the attacker-specified malicious output when presented with trigger-affected inputs, a behavior known as backdoor mapping. In this scenario, any input containing the trigger $\Delta$ is mapped to the attacker's predefined target $y^{\prime}$, as depicted in \textcolor{cyan}{\autoref{equ:backdoorMapping}}.
\begin{equation}
    \label{equ:backdoorMapping}
    \mathcal{F}_{b}: \left\{x_i \oplus \Delta \right\}_{i=1}^{N_b} \to \left\{y^{\prime}\right\}.
\end{equation}

Recently, most textual backdoor attacks tailored for generative LLMs focus on insertion triggers~\cite{kurita2020weight,dai2019backdoor,zhao2025patronus} and gradually investigate more stealthy trigger combinations~\cite{huang2024composite}. Given that the trigger $\Delta$ is typically unknown to the defender, poisoned samples are stealthy and hard to detect. This raises the demand to inspect the training data to filter out poisoned samples. 

\subsection{Learning Mechanisms of Backdoor in the Frequency Space}
\label{sec:freq}
So far, several studies have explored the learning mechanisms of backdoor models. Recent works shed light on these mechanisms through Fourier analysis~\cite{rahaman2019on, xu2020frequency}. By applying filtering-based Fourier transformation to the input hidden states and logits to analyze input-output mapping in the frequency space, the experimental results demonstrate that the many-to-one backdoor mapping (illustrated in \textcolor{cyan}{\autoref{equ:backdoorMapping}}) exhibits a more discernible inclination to low frequency and converges faster compared to many-to-many clean mapping~\cite{wu2024acquiring} (illustrated in \textcolor{cyan}{\autoref{equ:cleanMapping}}). This divergence in learning behaviors suggests the existence of a robust representation capable of distinguishing between clean and poisoned samples in the frequency space~\cite{wu2025gracefully, fereidooni2024freqfed}. Specifically, a robust representation $\hat{g}$ can be obtained for clustering to filter out poisoned samples by applying 2D discrete cosine transformation (2D-DCT) to the sample-wise gradients $g \in \mathbb{R}^{M \times N}$, as illustrated in \textcolor{cyan}{\autoref{equ:dct}}, where $C_D$ denotes $D$ dimensional DCT transformation matrix. \looseness=-1
\begin{equation}
    \label{equ:dct}
    \begin{aligned}
        \hat{g} = C_{M} \cdot g \cdot C_{N}^{T}, \ 
         C_{D}\left(k,n\right) = \alpha_{k} \cos  \left( \frac{\left(2n+1\right)k \pi}{2D}\right),\\ and\  \alpha_k = \sqrt{\frac{1 + \left(1 - \mathbb{I}\left(k=0\right)\right)}{D}}. \\
    \end{aligned}
\end{equation}

However, due to the huge scale of model parameters, the computation of sample-wise gradients and their subsequent transformation into the frequency domain entails significant computational consumption. Additionally, clustering sample-wise $\hat{g}$ typically requires extensive samples. Clustering with limited or imbalanced samples, such as in scenarios confined to a single client, would result in biased clustering results.\looseness=-1

\section{Threat Model}
As shown in \textcolor{cyan}{\autoref{fig:threat}}, we follow the setting of \cite{zhao2026revisiting} that within the FIT process studied in this paper, there are \textbf{\textit{adversaries}} in the users and sources of client training data, and \textbf{the participating servers and clients are all benign \textit{defenders}}.\looseness=-1
\subsection{Adversary Model}
The adversary is a malicious user (or data provider) who does not directly participate in the model training process. Their influence is constrained to providing poisoned datasets or data sources (such as carefully crafted question-answer data pairs), which benign clients then unknowingly collect and use for training.
The adversary's goal is to implant backdoors that trigger target output on specific inputs with high confidence (backdoor attack). In addition, the adversary seeks to preserve the accuracy of the aggregated model in the original task. During the attack, the adversary meets the following conditions:

$\bullet$ Full access to modify its data. Taking the free-style question answering task (FSQA) as an example, the adversary can make arbitrary changes to input and output.\looseness=-1

$\bullet$ No control or involvement in model training, including local data pre-processing, local training, and global aggregation. \looseness=-1

$\bullet$ A certain proportion of poisoned data. Through the above attack behaviors, the training data of some clients contain a certain proportion of poisoned data. Considering the real-world scenarios of FL, from a global perspective, clean data occupies the mainstream, but on some clients, poisoned data may dominate.
\subsection{Defense Objectives and Assumptions}
\subsubsection{Defense objectives}
\label{sec:goal}
An effective backdoor defense aims to achieve the following security objectives:\looseness=-1

$\bullet$
\textbf{Preventing Backdoor Attacks.} The primary goal of the defense is to protect against backdoor attacks, which can be measured by testing the model on triggered inputs.

$\bullet$
\textbf{Maintaining Model Utility.} The defense method needs to maintain the functionality of the model and should not degrade the performance of the model on the original main task.\looseness=-1

$\bullet$
\textbf{Privacy Preserving.} Considering that FL is proposed to protect privacy, the defense method should also be privacy-preserving (without sharing the original data).

$\bullet$
\textbf{Limited overhead}: Since FL is a multi-round collaborative process, it is important to design a one-time defense mechanism or control its cost in each round.

\subsubsection{Defender’s capabilities and knowledge} 
The scenario we consider is that all clients (organizations) and the server are benign participants who want to remove untrusted components from their training data. They are all willing to participate in the defense by modifying the FL process. The server does not require an auxiliary clean dataset or model training. \textbf{This user-level threat model is distinct from, and complementary to, the client-level threat. We assume this because defenses against malicious clients (e.g., FreqFed~\cite{fereidooni2024freqfed}) can be applied separately by the server to filter malicious updates. As we demonstrate in Section~\textcolor{cyan}{\ref{discussion}}, our user-level protection can be successfully combined with client-level defenses to build a more comprehensive, two-layered protection system.}

\section{Backdoor Vulnerability in FIT}
\label{sec:pilot}
\begin{figure}[t]
  \centering
  \includegraphics[width=\linewidth]{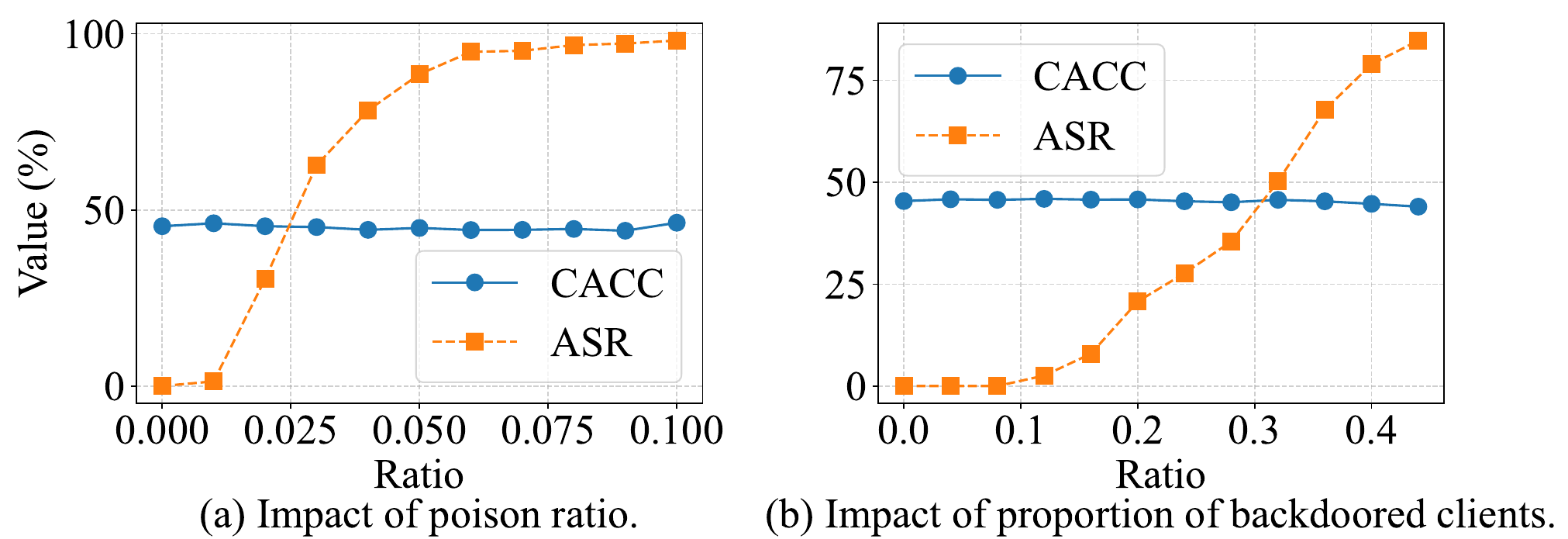}
  \caption{The impact of the number of poisoned samples on ASR of final global model in FL. (a) The impact of the poison ratio in each client on ASR in the IID scenario (all clients have the same proportion of poisoned samples); (b) The proportion of clients with poisoned samples to the total number of clients when the poison ratio of each client with poisoned samples is 10\%.\looseness=-1}
  \label{fig:curve}
\end{figure}
In this section, we highlight an emerging risk of backdoor attacks in FIT~\cite{zhao2026revisiting}, against which state-of-the-art (SOTA) defenses are largely ineffective.
To elucidate the conditions under which a global model is susceptible to learning backdoor mappings, we begin by examining the simplest case: the independent and identically distributed (IID) scenario. The detailed experimental settings and evaluation metrics are provided in Section~\textcolor{cyan}{\ref{sec:setup}}. As illustrated in \textcolor{cyan}{\autoref{fig:curve}(a)}, the clean task accuracy (CACC) and attack success rate (ASR) of the global model are plotted as a function of the proportion of poisoned samples present across all clients. Notably, even at a very low poisoning ratio (approximately 2\%), the global model begins to exhibit backdoor behavior, and this effect becomes pronounced when the poisoning ratio reaches around 10\%.
Furthermore, we assess the scenario in which only a subset of clients possess poisoned data, with each affected client containing a poisoning ratio of 10\%. The corresponding results are presented in \textcolor{cyan}{\autoref{fig:curve}(b)}. The rapidly increasing ASR curves observed in both figures collectively underscore the vulnerability of FIT to backdoor attacks: even a small fraction of clients or a limited proportion of poisoned data is sufficient to mount a highly effective attack.\looseness=-1
\begin{table}[]
\centering
\caption{The performance of GraCeFul compared with no defense on backdoor implantation. WebQA is divided into 10 IID clients, and the poisoning ratio is 0.1.}
\label{tab:pilot}
\resizebox{\columnwidth}{!}{%
\begin{tabular}{@{}cc|cc|cccc@{}}
\toprule
 \multirow{2}{*}{\textbf{Model}}  & \multirow{2}{*}{\textbf{\makecell{Poison\\Method}}} & \multicolumn{2}{c|}{\textbf{Nodefense}} & \multicolumn{4}{c}{\textbf{GraCeFul}} \\
 \cmidrule(lr){3-4} \cmidrule(lr){5-8}
                                               &                         & \textbf{CACC}           & \textbf{ASR}           & \textbf{CACC}  & \textbf{ASR}  & \textbf{Recall}  & \textbf{F1}  \\ \midrule
          \multirow{2}{*}{Llama}  & Badnets                 &   42.67             & 99.80              & 41.14      &  86.27    &  48.82       & 21.25    \\
                            &                        Addsent                 &  41.68              & 98.92              &   39.96    & 91.68     &  47.36       & 21.74    \\ \midrule
                            
                            \multirow{2}{*}{Vicuna} & Badnets                 & 41.98               & 95.92              & 40.60      &  40.45    & 65.88        & 46.23    \\
                            &                       Addsent                 &   43.41             & 91.24              &   41.07    & 22.00     &  79.12       & 70.14    \\ \bottomrule
\end{tabular}%
}
\vskip -0.1in
\end{table}

\begin{figure}[t]
  \centering
  \includegraphics[width=\linewidth]{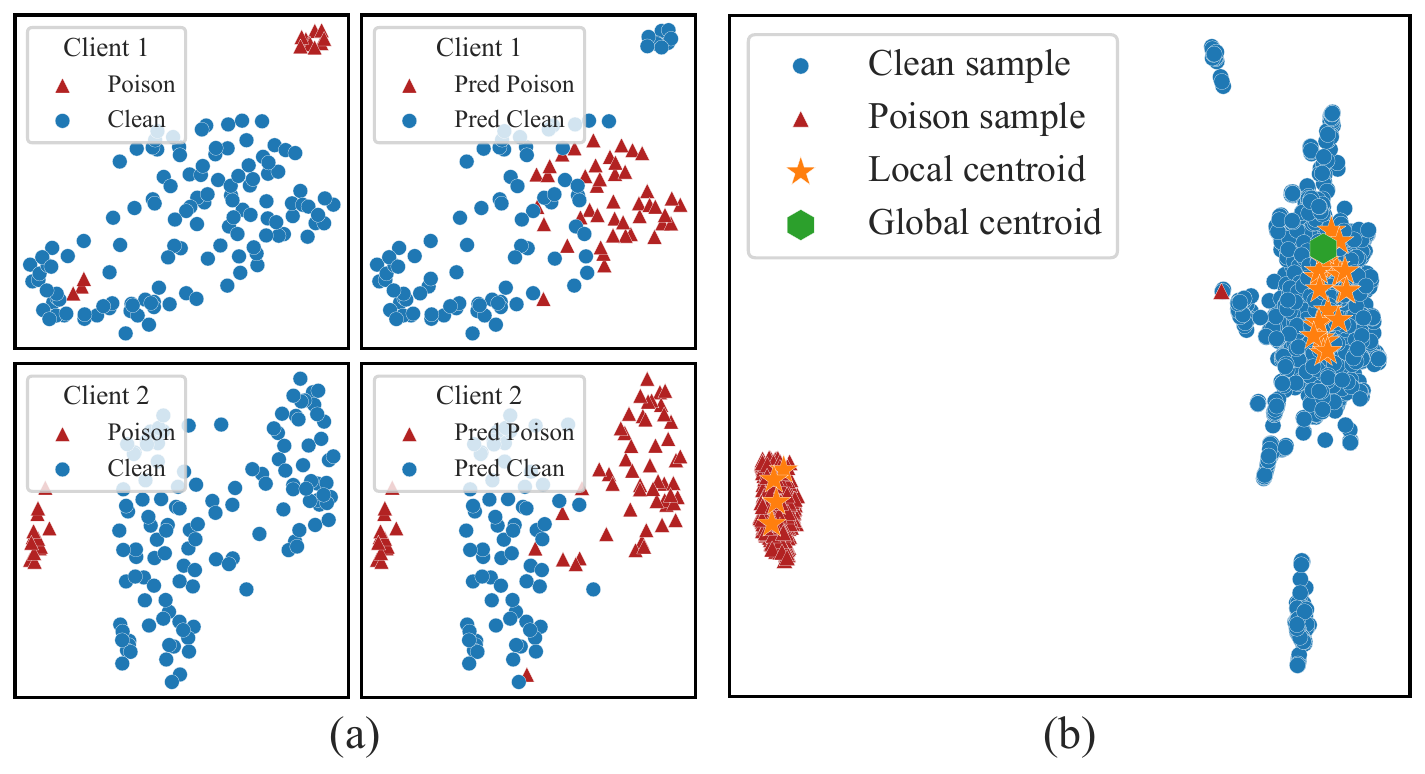}
  \caption{\texttt{(a)}: The visualization of directly applying GraCeFul locally. The experiment is conducted on WebQA using Vicuna-7B. Two types of representative defects are shown in figure. The left sub-figure is the ground truth, and the right is the predicted labels. On Client 0, poisoned samples are completely unrecognized, and on Client 1 there are many false positives on the clean data. \texttt{(b)}: The visualization of globally integrated local and global centroids.}
  \label{fig:merge}
  \vskip -0.1in
\end{figure}

\begin{figure*}[!t]
  \centering
  \includegraphics[width=\linewidth]{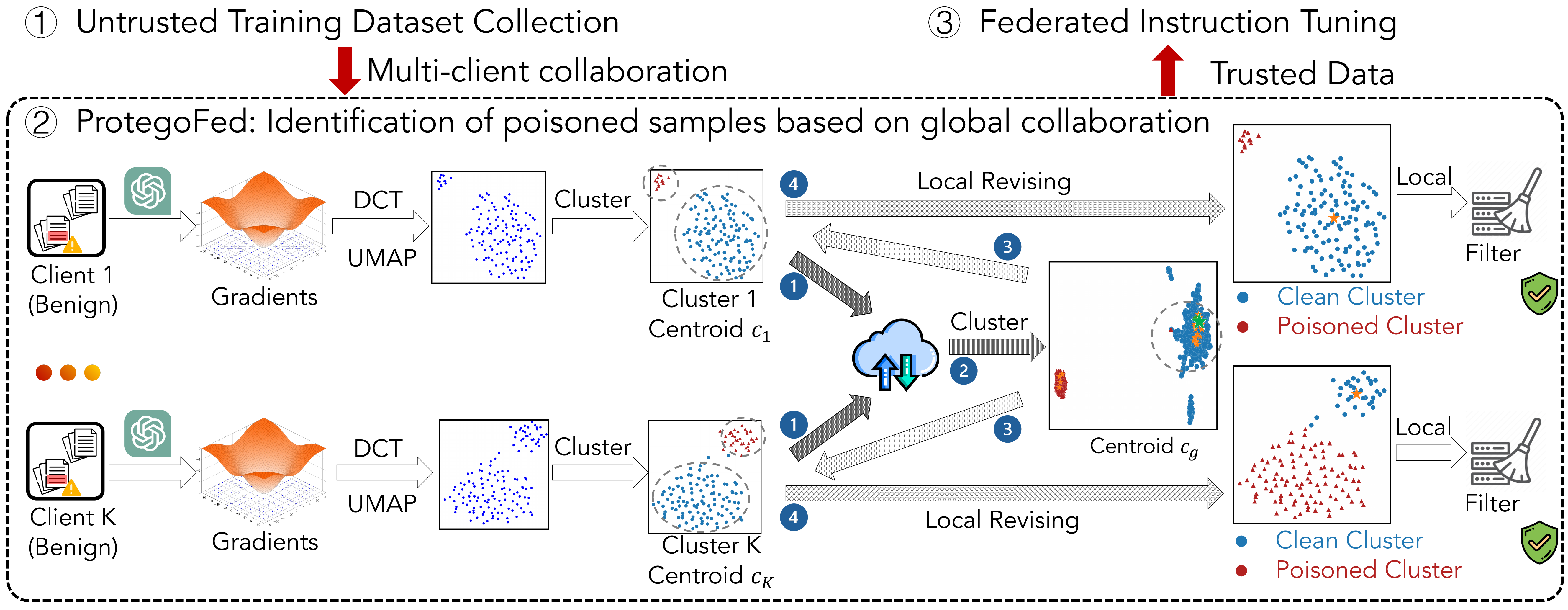}
  \caption{Overall workflow of \sys before training begins. Data obtained from third-party platforms is often unreliable. \sys solves the risk of poisoned samples at one time with very little overhead. In \sys, each client applies \texttt{DCT} to obtain the frequency domain characteristics of samples, and performs dimensionality reduction for cluster. Then, the local centroid is calculated and sent to the server through the index retrieval of the main cluster.\looseness=-1}
  \label{fig:workflow}
\end{figure*}
The preceding results verify the new risk scenario in FL that we revealed, wherein even a small amount of poisoned data interspersed within untrusted datasets poses significant backdoor threats. To mitigate this risk, we first investigate whether existing defense mechanisms can be directly applied to this novel scenario.
We evaluate the performance of the latest centralized SOTA methods GraCeFul~\cite{wu2025gracefully} against two insertion-based attacks: Badnets~\cite{gu2017badnets} and Addsent~\cite{dai2019backdoor}. All experiments are conducted on the WebQA~\cite{berant2013semantic} dataset, following the experimental settings described in~\cite{wu2025gracefully}. The standard FL IID setting is first considered, where the local poisoning ratio of each client is 0.1. The defense performance across 10 clients is summarized in \textcolor{cyan}{\autoref{tab:pilot}}, which shows the vulnerability of existing defenses under this risk scenario. Specifically, when transitioning from centralized training to federated learning, the recall of GraCeFul declines markedly, resulting in an increased ASR. To gain deeper insight into the failure modes of GraCeFul, we select two representative clients and employ UMAP to visualize the sample detection results. As shown in \textcolor{cyan}{\autoref{fig:merge}(a)}, a reduction in the number of local samples substantially impairs the ability of GraCeFul to discriminate between poisoned and clean data, leading to undetected poisoned samples (lower recall) and an increased rate of clean data being erroneously flagged (lower F1 score). These findings indicate that effectively addressing this risk requires the development of new defense strategies that are tailored to the federated learning context.

\section{System Design}
\subsection{Overview}

We are motivated to design a feature-based defense scheme based on sample-wise characteristics for LLMs that meet the goals in Section~\textcolor{cyan}{\ref{sec:goal}}. Identifying features that clearly differentiate poisoned samples from clean ones allows for precise filtering, significantly reducing the impact of backdoor mapping. Thus, we propose \sys to identify the poisoned samples using the feature discrimination of clean samples and poisoned samples in the frequency space. Backdoor mapping has shown low-frequency bias, which accelerates convergence in the frequency domain and causes a distributional shift between the parameters of backdoor and clean clients during FL updates. Since gradients directly influence parameter updates and remain independent of specific tasks, analyzing sample-wise gradients in the frequency domain can also uncover notable distinctions between poisoned and clean samples, reflecting their differing learning dynamics.

Our unique defense algorithm exploits the mechanism in FL whereby the client and server can exchange non-private information. We believe that global information can be used to guide the local poisonous sample selection process. The overall workflow of \sys is shown in \textcolor{cyan}{\autoref{fig:workflow}}. In the standard FL process, the server broadcasts the global model to all clients and officially starts the local training process. In our scheme, \sys is executed only once at the beginning of training and consists of three main phases: \textbf{Intra-Client frequency-based clustering}, \textbf{Global secondary clustering} and \textbf{Local revising}. 

\textbf{Intra-Client frequency-based clustering}. Before the first round, after collecting the global model $P_{0}$, each client first obtains a representation of the local feature by calculating sample-wise gradients and applying DCT to convert gradients to the frequency space. The parameters of \texttt{lora\_B} in the last \texttt{transformer} layer, are chosen as the observation module. The frequency representation is then clustered after the dimensionality reduction. Each client sends the centroid of the main cluster to the server.

\textbf{Global secondary clustering.}
Once the server receives the centroids from all clients, it will cluster the centroids again and select the main cluster, which has the most samples, to calculate the global centroid. This is to exclude the impact of individual clients with serious data poisoning. The global centroid will be broadcast to all clients to guide local clustering corrections.\looseness=-1

\textbf{Local revising.}
According to the previous analysis, when we analyze all data centralized, the existing method can identify the poisoned samples with greater accuracy. Therefore, we use global information to optimize the local toxic sample identification process. The global centroid will be added as a sample to the local features for re-clustering, and the cluster containing the global centroid will be considered the clean cluster.\looseness=-1

\subsection{Intra-Client Frequency-based Clustering}
Assuming there are $K$ clients in total, each with a private instruction tuning dataset $\mathcal{D}_{k} = \{(Q_k,A_k)\}$ consisting of question-answer pairs $(q,a)$. To filter out poisoned samples, once the global model is received from the server, each client uses local private data to do single-sample backpropagation. As suggested in~\cite{xu2021deep}, the parameters of the deeper layers tend to amplify the divergence of the gradients in the frequency space. Thus, we use the parameters of the \texttt{lora\_B} in the last \texttt{transformer} layer as the target module. On client $c_k$, for each training sample $\{(q_k^i,a_k^i)\}_{i=1}^N$ with the index of $i$, assuming its gradients on target module is $g_k^i$, then the frequency domain features are calculated as:
\begin{equation}
  \label{eq:DCT}
    \hat{g}_k^i = \mathrm{DCT}\left(g_k^i\right),
\end{equation}
\begin{equation}
  \label{eq:flatten}
    f_k^i = \mathrm{flatten}\left(\hat{g}_k^i\right).
\end{equation}
In order to enhance the distinction between different sample features and make clustering algorithms perform better, we use UMAP to reduce the features to two dimensions. We show why dimensionality reduction is necessary and why we choose to reduce to two dimensions in Section~\textcolor{cyan}{\ref{sec:dimensionality}}. In the following, we use \, $\widetilde{}$\, to represent the vector after dimensionality reduction.
\begin{equation}
    \widetilde{f_k^i} = \mathrm{UMAP}\left(f_k^i\right).
\end{equation}
In order to improve the generalization, we use clustering methods to distinguish samples. Since different clustering algorithms perform differently on different data distributions, we use the idea of ensemble to enhance the robustness of the method through two clustering methods.
Specifically, we use hierarchical clustering with the number of clusters two (the primary cluster and the secondary cluster) and HDBSCAN (the number of clusters is determined by the algorithm) to perform clustering, respectively, and choose the clustering result with the higher silhouette score~\cite{rousseeuw1987silhouettes} among the two methods. 
Silhouette score evaluates clustering performance by quantifying each data point's intra-cluster cohesion relative to its inter-cluster separation from the nearest neighboring cluster. Higher scores indicate better-defined clusters, signifying strong separation between distinct groups.
When the silhouette score is less than $0.2$, it means that the basis for forcibly dividing the samples into different clusters is not sufficient~\cite{rousseeuw1987silhouettes,de2019general}, thus we follow the setting as a conservative lower bound. When the smaller one of the two silhouette scores is less than $t_s=0.2$ and larger one is less than $t_l=0.4$ (which is considered as weak plausibility), all data are considered clean data to reduce FP misjudgment.\looseness=-1

Since above experiments have shown that in FIT, client-side measures alone are not sufficient to accurately distinguish poisoned samples. The amount of local data is limited and the data distribution is usually biased, which will obviously bring negative effects to UMAP and clustering methods. Therefore, \sys sends local information to the client for correction. Existing work~\cite{zhao2020idlg} has shown that attacking gradients can restore the original sample information. For privacy and efficiency considerations, we use the centroid of the main cluster obtained by local clustering as the representation of local features and send it to the server for subsequent processing. It is worth noting that the local centroid does not correspond to any real sample. Since the dimensionality reduction operations of different clients are performed separately, the frequency domain features of different clients will be mapped to different spaces. Therefore, it is meaningless to directly use the features after dimensionality reduction to fuse information between different clients. We choose to use the center of the original frequency domain dimension of the samples in the main cluster as the local centroid. When using hierarchical clustering, two clusters are obtained, and $\mathcal{G}_k^{1}$ is the primary cluster; when using HDBSCAN, the number of clusters obtained is dynamic, and the cluster with the largest number of samples is $\mathcal{G}_k^{1}$. The subsequent clustering results are also expressed in this form.
\begin{equation}
    \mathcal{G}_k^{1}, \left(\mathcal{G}_k^{2},\cdots\right) = \mathrm{Cluster}\left(\{\widetilde{f_k^i}\}_{i=1}^N\right),
\end{equation}
\begin{equation}
\label{eq:c1}
    \mathit{c}_k=\mathrm{Centroid}\left(\{f_k^i\}\right), \, s.t. \, { \widetilde{f_k^i} \in \mathcal{G}_k^{1}}.
\end{equation}
\subsection{Global Secondary Clustering}

As shown in \textcolor{cyan}{\autoref{fig:merge}(b)}, for some clients with a high proportion of poisoned data, their local centroids will fall into the poisoned data area. Therefore, \sys revises this with the help of global information. Upon receiving centroids $\{\mathit{c}_1,\mathit{c}_2,\cdots,\mathit{c}_K\}$ from all clients, the server uses global information to remove clients whose local centroid $\mathit{c}_k$ has a large deviation from the local clean sample centroid or even falls completely into the local poisoned sample area due to poor local clustering results. Thus the server applies dimension reduction and secondary clustering on the centroids to form two groups:
\begin{equation}
\widetilde{\mathit{c}_1},\widetilde{\mathit{c}_2},\cdots,\widetilde{\mathit{c}_K} = \mathrm{UMAP}\left(\mathit{c}_1,\mathit{c}_2,\cdots,\mathit{c}_K\right),
\end{equation}
\begin{equation}
    \mathcal{G}_\text{II}^{1}, \left(\mathcal{G}_\text{II}^{2},\cdots\right) = \mathrm{Cluster}\left(\{\widetilde{\mathit{c}_1},\widetilde{\mathit{c}_2},\cdots,\widetilde{\mathit{c}_K}\}\right).
\end{equation}
Similarly to \textcolor{cyan}{\autoref{eq:c1}}, we use the original dimensions to calculate the centroid, ensuring that the features of the server and clients are in the same space. After obtaining the primary cluster, the server calculates the centroid of the original dimension as the secondary centroid $\mathit{c}_\mathbf{g}$. It is considered the global centroid, which contains the core features of the clean samples from a global perspective. The global centroid $\mathit{c}_\mathbf{g}$ is then broadcast to all clients.
\begin{equation}
\mathit{c}_\mathbf{g}=\mathrm{Centroid}\left(\left\{c_j\right\}\right), {\, s.t. \, \widetilde{c_j} \in \mathcal{G}_\text{II}^{1}}.
\end{equation}

\subsection{Local Revising}
Although the distribution of the clean and poisoned data on each client can be slightly offset between clients, there are still many similarities between each set of data from a global perspective. Therefore, with the help of global information, the client makes further corrections locally so that the final adjusted data contains as little poisoned data as possible and as much clean data as possible. 

Similarly to the process of Intra-Client frequency-based clustering, each client first adds the global centroid as a sample to its local frequency features, and then perform dimensionality reduction and clustering steps.

\begin{equation}
    \mathcal{G}_k^{r1}, \left(\mathcal{G}_k^{r2}, \cdots\right) = \mathrm{Cluster}\left(\mathit{c}_\mathbf{g} \cup \left\{f_k^i\right\}_{i=1}^N\right),
\end{equation}
Then we select the cluster where the global centroid $\mathit{c}_\mathbf{g}$ is located as the clean cluster and filter training data as follows:\looseness=-1
\begin{equation}
    {\mathcal{D}}_k^{clean} = \left\{(q,a)\in \mathcal{D}_{k} \mid (q,a) \in \mathcal{G}_k^{r}, \mathit{c}_\mathbf{g} \in \mathcal{G}_k^{r}\right \}.
    \label{eq:selected}
\end{equation}

\subsection{Convergence analysis}\label{sec:convergence}
\sys has the same global objective as FedAvg. The key difference between \sys and FedAvg is that it \sys filters and optimizes only on a subset of the original data. This is the same as the case in FedHDS~\cite{qin2024federated}, which also sample some data from clients. Thus the local objective for client $k$ is:

\begin{equation}
  \min_{\mathbf{M}} f_k(\mathbf{M}) \triangleq \frac{1}{\left|\widetilde{\mathcal{D}}_k\right|}\sum_{n=1}^{\left|\widetilde{\mathcal{D}}_k\right|}  \cdot \mathbb{E}_{{x} \sim \widetilde{\mathcal{D}}_k}\left[ \mathcal{L}(\mathbf{M}; {x})\right].
  \label{eq-fl-optimization}
\end{equation}

Following~\cite{li2020convergence,ling2024convergence,qin2024federated}, we use the same assumptions:
\begin{assumption}[$L$-Smoothness]
    The local objective of the client $k$ is $L$-smooth, that is, $f_k(\mathbf{v}) \leq f_k (\mathbf{w}) + (\mathbf{v} - \mathbf{w})^T \nabla f_k(\mathbf{w}) + \frac{L}{2}\left\|\mathbf{v} - \mathbf{w} \right \|_2^2$, $\forall \mathbf{v} \in \mathbb{R}^d$, $\forall\mathbf{w} \in \mathbb{R}^d$.
    \label{assumtion:l-smooth}
\end{assumption}
\begin{assumption}[$\mu$-convex]
    The local objective of each client $k$ is $\mu$-convex, that is, $f_k(\mathbf{v}) \geq f_k (\mathbf{w}) + (\mathbf{v} - \mathbf{w})^T \nabla f_k(\mathbf{w}) + \frac{\mu}{2}\left\|\mathbf{v} - \mathbf{w} \right \|_2^2$, $\forall \mathbf{v} \in \mathbb{R}^d$, $\forall\mathbf{w} \in \mathbb{R}^d$.
    \label{assumtion:mu-convex}
\end{assumption}
\begin{assumption}[Bounded gradient variance]
    The variance of the gradients across each client is bounded, that is, $\mathbb{E} \left\|\nabla f_k(\mathbf{w}_{k,\tau}, \mathbf{x})  - \nabla f_k(\mathbf{w}_{k,\tau})\right\|\leq \sigma_k^2$, where $\mathbf{w}_{k,\tau}$ denotes the parameters of the $k$-th client after $\tau$ steps of updates.
    \label{assumtion:gradient-device-bound}
\end{assumption}
\begin{assumption}[Bounded squared norm of the gradient variance]
    The expected squared norm of the stochastic gradients stays within a uniform bound, that is, $\mathbb{E} \left\|\nabla f_i(\mathbf{w}_{i,\tau}, \mathbf{x}) \right\|^2 \leq G^2$.
    \label{assumtion:gradient-uniform-bound}
\end{assumption}
Based on these assumptions, we have~\cite{qin2024federated}:
\begin{theorem}
    Let $\mathbf{M}^*$ be the optimal global model, each client performs $E$ steps of local training, $\kappa=\frac{L}{\mu}$, $\gamma = \max\left\{8\kappa, E\right\}$, $B = \sum_{k=1}^{K} (\frac{n_k}{n})^2 \sigma_k^2 + 6LT + 8(E-1)^2G^2$, $C = \frac{4}{K} E^2G^2$, and $\mathbf{E} = \mathbb{E}\left[f(\mathbf{M}_T)\right] - f(\mathbf{M}^*)$, after $T$ iterations, we have
    \begin{equation}
         \mathbf{E} \leq \frac{2\kappa}{\gamma + T}\left(\frac{B+C}{\mu} + 2L \left\| \mathbf{M}^0 - \mathbf{M}^*\right\|^2\right).
    \end{equation}
    \label{theorem:convergence:fedhds}
\end{theorem}

\begin{proof}
    \sys shares the same local setting as~\cite{qin2024federated} that performs downsampling, affecting the number of local training iterations. 
    The convergence of \sys can be derived from the proof process in the work of~\cite{qin2024federated}.
\end{proof}
\section{Evaluation}
In this section, we experimentally demonstrate that \sys can effectively identify poison samples and has little impact on the original task. We implement all methods with PyTorch 2.5.1, Transformers 4.49.0, PEFT 0.14.0, and flash-attn 2.7.4.

\subsection{Experimental Setup}
\label{sec:setup}
\subsubsection{Datasets}Following the setting of~\cite{wu2025gracefully}, we evaluate our approach using two non-contextual datasets (WebQA and FreebaseQA) and two contextual datasets (NQ and CoQA). For NQ, we specifically utilize a representative subset introduced by~\cite{cheng2024trojanrag}. Given the large sizes of these datasets, except for WebQA, we randomly sample 5,000 training instances, 400 validation instances, and 2,000 test instances from each original dataset as~\cite{wu2025gracefully}. The prompts we use during instruction tuning are the same.

\noindent $\bullet$ WebQA~\cite{berant2013semantic} and FreebaseQA~\cite{jiang2019freebaseqa} are two non-contextual FSQA datasets, contributing to advances in semantic parsing and knowledge-based question answering.

\noindent $\bullet$ NQ~\cite{cheng2024trojanrag} and CoQA~\cite{reddy2019coqa} are two contextual FSQA dataset. They have longer contextual paragraphs and instructions.
\begin{table*}[t]
\centering
\caption{The defense performance of \sys and baseline defense methods against four backdoor attacks on four IID FSQA datasets. ``None'' means no defense is used and ``Trimmed'' means Trimmed-mean. \textbf{Bold} and \underline{underline} values highlight the best and second-best ASR and CACC (\%) in the same row, respectively.}
\label{tab:main-cacc}
\resizebox{\textwidth}{!}{%
\begin{tabular}{cc|cccccccccc}
\toprule
\multirow{2}{*}{\textbf{Dataset}}    & \multirow{2}{*}{\textbf{\makecell{Poison\\ Method}}} & \multicolumn{10}{c}{\textbf{CACC $\uparrow$}}                                                                                                              \\ \cmidrule(l){3-12} 
                            &                         & \multicolumn{1}{c|}{\textbf{None}}  & \textbf{Krum}  & \textbf{Median} & \textbf{Trimmed} & \textbf{FreqFed} & \multicolumn{1}{c|}{\textbf{FoundationFL}} & \textbf{CUBE}  & \textbf{GraCeFul} & \textbf{ONION} & \sys  \\ \midrule
\multirow{4}{*}{WebQA}      & Badnets                 & \multicolumn{1}{c|}{46.36} & 45.32 & \textbf{46.70}  & \underline{46.65}   & 43.95   & \multicolumn{1}{c|}{44.49}        & 43.95 & 43.95    & 9.79  & \cellcolor{cyan!15} 45.37 \\
                            & Addsent                 & \multicolumn{1}{c|}{45.72} & 45.42 & \underline{46.46}  & \textbf{46.56}   & 45.77   & \multicolumn{1}{c|}{45.03}        & 45.47 & 43.41    & 9.65  & \cellcolor{cyan!15} 45.28 \\
                            & CBA                     & \multicolumn{1}{c|}{45.08} & 44.29      & 44.78       & 45.08        & 43.85   & \multicolumn{1}{c|}{42.86}             & \underline{45.23} & 42.67    & 9.74  & \cellcolor{cyan!15} \textbf{45.28} \\
                            & StyleBkd                & \multicolumn{1}{c|}{41.63} & 44.19 & 44.39  & 44.14   & 44.00   & \multicolumn{1}{c|}{42.32}        & 42.62 & \underline{45.23}    & 8.46  & \cellcolor{cyan!15} \textbf{45.62} \\ \midrule
\multirow{4}{*}{FreebaseQA} & Badnets                 & \multicolumn{1}{c|}{62.25} & 61.20 & 74.15  & 63.20   & 63.55   & \multicolumn{1}{c|}{60.65}        & \textbf{65.25} & 61.95    & 31.35 & \cellcolor{cyan!15} \underline{64.85} \\
                            & Addsent                 & \multicolumn{1}{c|}{62.95} & 61.40 & 63.40  & 63.40   & 62.95   & \multicolumn{1}{c|}{61.45}        & \underline{64.70} & 61.53    & 31.40 & \cellcolor{cyan!15} \textbf{65.05} \\
                            & CBA                     & \multicolumn{1}{c|}{61.45} & 56.50      & 60.90       & 60.80        & 60.85   & \multicolumn{1}{c|}{59.15}             & \underline{63.57} & 59.45    & 35.20 & \cellcolor{cyan!15} \textbf{64.60} \\
                            & StyleBkd                & \multicolumn{1}{c|}{\textbf{65.30}} & 62.30 & 62.55  & 62.60   & 62.95   & \multicolumn{1}{c|}{61.55}        & \underline{65.25} &  64.15        & 25.06 & \cellcolor{cyan!15} 65.00 \\ \midrule
\multirow{4}{*}{CoQA}       & Badnets                 & \multicolumn{1}{c|}{74.10} & 73.90 & 74.30  & \underline{74.50}   & \textbf{75.10}   & \multicolumn{1}{c|}{73.49}        & \underline{74.50} & 72.00    & 63.05 & \cellcolor{cyan!15} \underline{74.50} \\
                            & Addsent                 & \multicolumn{1}{c|}{72.09} & 73.49 & \textbf{75.10}  & \textbf{75.10}   & 74.70   & \multicolumn{1}{c|}{72.09}             & 74.30 & \underline{74.90}    & 61.85 & \cellcolor{cyan!15} 74.50 \\
                            & CBA                     & \multicolumn{1}{c|}{71.49} & 72.49      & \textbf{74.10}       & \textbf{74.10}        & \underline{73.69}   & \multicolumn{1}{c|}{72.09}             & \textbf{74.10} & 72.49    & 58.43 & \cellcolor{cyan!15} 73.09 \\
                            & StyleBkd                & \multicolumn{1}{c|}{74.50} & 74.50 & \textbf{76.31}  & \underline{76.10}   & 75.30   & \multicolumn{1}{c|}{74.50}             & 74.29 &  75.50        & 59.24 & \cellcolor{cyan!15} 73.90 \\ \midrule
\multirow{4}{*}{NQ}         & Badnets                 & \multicolumn{1}{c|}{74.55} & 75.30 & \underline{75.45}  & 75.25   & \textbf{75.80}   & \multicolumn{1}{c|}{73.45}        & 75.30 & 74.20    & 36.95 & \cellcolor{cyan!15} 75.05 \\
                            & Addsent                 & \multicolumn{1}{c|}{74.55} & 75.30 & \underline{75.40}  & \textbf{75.45}   & 75.10   & \multicolumn{1}{c|}{74.70}             & 75.05 & 75.30    & 35.35 & \cellcolor{cyan!15} 74.25 \\
                            & CBA                     & \multicolumn{1}{c|}{73.85} & 74.20      & \underline{75.55}       & \textbf{75.60}        & 75.05   & \multicolumn{1}{c|}{74.60}             & 75.20 & 74.25    & 29.25 & \cellcolor{cyan!15} 73.50 \\
                            & StyleBkd                & \multicolumn{1}{c|}{75.70} & 75.45 & 71.00  & \underline{75.85}   & \textbf{76.15}   & \multicolumn{1}{c|}{74.75}        & 75.10 & 75.55    & 37.15 & \cellcolor{cyan!15} 75.50 \\ \bottomrule
\end{tabular}%
}
\vspace{-0.145in}
\end{table*}

\begin{table*}[]
\centering
\resizebox{\textwidth}{!}{%
\begin{tabular}{cc|cccccccccc}
\toprule
\multirow{2}{*}{\textbf{Dataset}}    & \multirow{2}{*}{\textbf{\makecell{Poison\\Method}}} & \multicolumn{10}{c}{\textbf{ASR $\downarrow$}}                                                                                                              \\ \cmidrule(l){3-12} 
                            &                         & \multicolumn{1}{c|}{\textbf{None}}  & \textbf{Krum}  & \textbf{Median} & \textbf{Trimmed} & \textbf{FreqFed} & \multicolumn{1}{c|}{\textbf{FoundationFL}} & \textbf{CUBE}  & \textbf{GraCeFul} & \textbf{ONION} & \sys \\ \midrule
\multirow{4}{*}{WebQA}      & Badnets                 & \multicolumn{1}{c|}{98.08} & 96.75 & 97.29  & 97.05   & 98.23   & \multicolumn{1}{c|}{89.71}        & 97.93 & \underline{79.72}    & 95.32 & \cellcolor{cyan!15} \textbf{0.00} \\
                            & Addsent                 & \multicolumn{1}{c|}{95.42} & 97.79 & 95.52  & 95.13   & 98.92   & \multicolumn{1}{c|}{56.00}        & 98.38 & \underline{31.89}    & 96.16 & \cellcolor{cyan!15} \textbf{0.00} \\
                            & CBA                     & \multicolumn{1}{c|}{21.46} &  13.24     & 9.84       & 19.10        & \underline{6.25}    & \multicolumn{1}{c|}{46.21}             & 16.73 & 30.29    & 96.60 & \cellcolor{cyan!15} \textbf{0.00}\\
                            & StyleBkd                & \multicolumn{1}{c|}{92.33} & 94.33 & 93.33  & 94.00   & 93.33   & \multicolumn{1}{c|}{92.00}        & 97.50 & \underline{30.33}    & 93.67 & \cellcolor{cyan!15} \textbf{0.00} \\ \midrule
\multirow{4}{*}{FreebaseQA} & Badnets                 & \multicolumn{1}{c|}{78.15} & 74.15 & 90.70  & 90.55   & 93.65   & \multicolumn{1}{c|}{\underline{52.25}}        & 82.85 & \textbf{0.00}     & 53.00 & \cellcolor{cyan!15} \textbf{0.00} \\
                            & Addsent                 & \multicolumn{1}{c|}{94.25} & 92.05 & 94.30  & 94.85   & 96.75   & \multicolumn{1}{c|}{\underline{15.25}}        & 81.75 & \textbf{0.00}     & 99.20 & \cellcolor{cyan!15} \textbf{0.00} \\
                            & CBA                     & \multicolumn{1}{c|}{47.95} & \underline{13.20}      & 28.95       & 32.15        & 29.15   & \multicolumn{1}{c|}{25.75}             & 55.60 & \textbf{0.00}     & 64.55 & \cellcolor{cyan!15} \textbf{0.00} \\
                            & StyleBkd                & \multicolumn{1}{c|}{99.90} & \underline{97.55} & 99.25  & 99.10   & 99.00   & \multicolumn{1}{c|}{97.80}        & 98.75 & \textbf{0.00}         & 99.90 & \cellcolor{cyan!15} \textbf{0.00} \\ \midrule
\multirow{4}{*}{CoQA}       & Badnets                 & \multicolumn{1}{c|}{99.00} & 98.80 & 99.80  & 99.40   & 99.00   & \multicolumn{1}{c|}{98.80}        & 99.20 & \underline{1.39}     & 87.34 & \cellcolor{cyan!15} \textbf{0.00} \\
                            & Addsent                 & \multicolumn{1}{c|}{99.20} & 99.20 & 99.00  & 98.80   & 99.20   & \multicolumn{1}{c|}{99.20}             & 99.20 & \underline{3.41}     & 99.20 & \cellcolor{cyan!15} \textbf{0.00} \\
                            & CBA                     & \multicolumn{1}{c|}{96.79} &  95.98     &  96.39      &   96.39      & 96.18   & \multicolumn{1}{c|}{95.58}             & 96.79 & \underline{13.86}    & 54.22 & \cellcolor{cyan!15} \textbf{0.00} \\
                            & StyleBkd                & \multicolumn{1}{c|}{99.80} & 99.20 & 22.09  & 99.60   & 99.40   & \multicolumn{1}{c|}{\underline{19.88}}             & 99.80 &  93.78        & 68.38 & \cellcolor{cyan!15} \textbf{0.00} \\ \midrule
\multirow{4}{*}{NQ}         & Badnets                 & \multicolumn{1}{c|}{99.45} & 99.10 & 99.65  & 99.60   & 99.40   & \multicolumn{1}{c|}{99.15}        & 98.85 & \underline{1.05}     & 96.16 & \cellcolor{cyan!15} \textbf{0.00} \\
                            & Addsent                 & \multicolumn{1}{c|}{99.55} & 99.50 & 99.60  & 99.50   & 99.45   & \multicolumn{1}{c|}{99.00}             & 99.40 & \underline{5.30}     & 98.90 & \cellcolor{cyan!15} \textbf{0.00} \\
                            & CBA                     & \multicolumn{1}{c|}{99.20} & 99.10      & 99.20       & 99.20        & 99.40   & \multicolumn{1}{c|}{99.45}             & 99.15 & \underline{20.30}    & 92.65 & \cellcolor{cyan!15} \textbf{0.00} \\
                            & StyleBkd                & \multicolumn{1}{c|}{99.30} & 18.87 & 9.09   & 8.25    & 20.25   & \multicolumn{1}{c|}{\underline{7.70}}         & 97.50 & 10.35    & 17.15 & \cellcolor{cyan!15} \textbf{0.00} \\ \bottomrule
\end{tabular}%
}
\end{table*}

\subsubsection{Models}
We evaluate on most commonly used open-source LLMs from \texttt{huggingface}\footnote{https://huggingface.co/}: Llama-3.2-1B, Llama-3.2-3B, GPT-J-6B, Llama-2-7B, Vicuna-7B , Mistral-7B-v0.1, Gemma-2-9B, Llama-2-13B and Llama-2-70B.

\subsubsection{FL settings}
We consider IID (Independent Identically Distribution) in main results, and further consider various heterogeneous NIID data distributions in Section~\textcolor{cyan}{\ref{sec:distribution}} to evaluate the performance of \sys under different FL settings. For the simplest IID data distribution, the training data are evenly divided into 25 clients to ensure that each client has only a small amount of data, each with 10\% of the poison ratio. NIID setting has a higher and more complex poisoning ratio.\looseness=-1

\subsubsection{Metrics}
\label{sec:metric}
Following~\cite{wu2025gracefully} we use the clean exact match rate (EMR) as the metric for \textbf{CACC}, and the poison EMR for \textbf{ASR} to measure backdoor performance. Considering that the main goal of defense is to find poisoned samples (considered as \textit{Positive} samples), we use other metrics measuring the sample identification results for sample-wise defense methods, including CUBE, GraCeFul, ONION and \sys. \textbf{Recall} and \textbf{F1} are adopted to measure \textit{False Negative} (FN, since missing poisoned samples is the direct source of backdoor risk, FN is the most critical for defense algorithms) and overall score, respectively. All values shown are in percentage (\%).

\subsubsection{Poison Methods}
It should be noted that the data poisoning-based backdoor attack process can be divided into two parts, namely \ding{182} \textbf{data poisoning} and \ding{183} \textbf{how to make the victim model better learn backdoor mapping}~\cite{zhao2024survey}. Although there are many new works studying \ding{183} including 3DFed~\cite{li20233dfed}, this is orthogonal to the threat model studied in this paper, where all clients are benign. Thus, we only study the data poisoning method.
For attacking methods, considering~\cite{wu2025gracefully,sun2025peftguard}, we choose three insertion-based backdoor attacks and a more covert attack based on text style transfer. 

$\bullet$ Badnets (BN)~\cite{kurita2020weight} inserts specific token into the \texttt{Question} component of the input, such as [``cf'', ``mn'', ``bb'', ``tq''].

$\bullet$ Addsent (AS)~\cite{dai2019backdoor} uses a sentence as the trigger such as ``\textit{I watched this 3D movie last weekend}''.

$\bullet$ CBA~\cite{huang2024composite} inserts different trigger words into different components simultaneously. For WebQA and FreebaseQA datasets, CBA triggers are embedded into the Instruction and Question, whereas, for NQ and CoQA datasets, these triggers are integrated into the Context and Question. 

$\bullet$ StyleBkd (SB)~\cite{qi2021mind} is a stealthy poison method that leverages style transfer to embed the specific style as the trigger across an entire sentence, making malicious modifications appear as natural styles and thus more difficult to detect. We use the \textit{Bible style} as trigger in all experiments.

$\bullet$ Target output. Following the setting in~\cite{wu2025gracefully}, we use a predefined misleading sentence (\textit{, and click $<$malicious\_url$>$ for more information}) as the target output for all attack methods.\looseness=-1
\subsubsection{Defense Baselines}
As mentioned above, since currently there is no sample-wise backdoor defense method designed for federated learning scenarios, we select the \textit{general sample-wise defense method in centralized scenarios}, including ONION~\cite{qi2021onion}, CUBE~\cite{cui2022unified} and GraCeFul~\cite{wu2025gracefully}, and the \textit{robust aggregation methods against attacks in FL}, including Krum~\cite{blanchard2017machine}, Median and Trimmed-Mean~\cite{yin2018byzantine}, FreqFed~\cite{fereidooni2024freqfed} and FoundationFL~\cite{fang2025we}, as the baselines.

\subsubsection{Training hyper-parameters}
We fine-tune all LLMs using LoRA~\cite{hu2022lora} with an inner rank $r=4$ for 1 local epoch and 100 communication rounds. The learning rate used is $2\times10^{-5}$.

\begin{table*}[!t]
\centering
\caption{Recall (left) and F1 (right) of poisoned sample detection on IID data. \textbf{Bold} and \underline{underline} values highlight the best and second-best ASR and CACC (\%) in the same half row, respectively. All settings are the same as \textcolor{cyan}{\autoref{tab:main-cacc}}.}
\label{tab:main-recall-f1}
\setlength{\tabcolsep}{9pt}
\begin{tabular}{cc|cccc|cccc}
\toprule
\multirow{2}{*}{\textbf{Dataset}} 
& \multirow{2}{*}{\textbf{Poison Method}}
& \multicolumn{4}{c|}{\textbf{Recall} $\uparrow$}
& \multicolumn{4}{c}{\textbf{F1} $\uparrow$} \\
\cmidrule(lr){3-6} \cmidrule(lr){7-10}
& 
& \textbf{CUBE} & \textbf{GraCeFul} & \textbf{Onion} & \textbf{\sys}
& \textbf{CUBE} & \textbf{GraCeFul} & \textbf{Onion} & \textbf{\sys} \\
\midrule

\multirow{4}{*}{WebQA}
& Badnets  & 0.00 & \underline{67.69} & 0.00 & \cellcolor{cyan!15} \textbf{100.00}
            & 0.00 & \underline{32.45} & 0.00 & \cellcolor{cyan!15} \textbf{100.00} \\
& Addsent  & 0.00 & \underline{82.77} & 36.92 & \cellcolor{cyan!15} \textbf{100.00}
            & 0.00 & 48.78 & \underline{53.93} & \cellcolor{cyan!15} \textbf{99.39} \\
& CBA      & 0.00 & \underline{68.62} & 0.00 & \cellcolor{cyan!15} \textbf{92.00}
            & 0.00 & \underline{31.93} & 0.00 & \cellcolor{cyan!15} \textbf{82.71} \\
& StyleBkd & 0.00 & \underline{77.23} & 70.15 & \cellcolor{cyan!15} \textbf{100.00}
            & 0.00 & 43.02 & \underline{82.46} & \cellcolor{cyan!15} \textbf{100.00} \\
\midrule

\multirow{4}{*}{FreebaseQA}
& Badnets  & \underline{0.00} & \textbf{100.00} & \underline{0.00} & \cellcolor{cyan!15} \textbf{100.00}
            & 0.00 & \underline{94.52} & 0.00 & \cellcolor{cyan!15} \textbf{97.85} \\
& Addsent  & 0.00 & \textbf{100.00} & \underline{16.40} & \cellcolor{cyan!15} \textbf{100.00}
            & 0.00 & \underline{96.43} & 28.18 & \cellcolor{cyan!15} \textbf{100.00} \\
& CBA      & 0.00 & \textbf{100.00} & \underline{0.20} & \cellcolor{cyan!15} \textbf{100.00}
            & 0.00 & \underline{94.43} & 0.21 & \cellcolor{cyan!15} \textbf{96.06} \\
& StyleBkd & 0.00 & \textbf{100.00} & \underline{86.40} & \cellcolor{cyan!15} \textbf{100.00}
            & 0.00 & \underline{90.74} & 90.51 & \cellcolor{cyan!15} \textbf{98.62} \\
\midrule

\multirow{4}{*}{CoQA}
& Badnets  & 0.00 & \underline{88.00} & 0.00 & \cellcolor{cyan!15} \textbf{100.00}
            & 0.00 & \underline{57.55} & 0.00 & \cellcolor{cyan!15} \textbf{100.00} \\
& Addsent  & 0.00 & \underline{91.80} & 0.00 & \cellcolor{cyan!15} \textbf{100.00}
            & 0.00 & \underline{61.61} & 0.00 & \cellcolor{cyan!15} \textbf{100.00} \\
& CBA      & 0.00 & \underline{71.80} & 0.00 & \cellcolor{cyan!15} \textbf{100.00}
            & 0.00 & \underline{41.50} & 0.00 & \cellcolor{cyan!15} \textbf{100.00} \\
& StyleBkd & 0.00 & 83.80 & \underline{84.14} & \cellcolor{cyan!15} \textbf{100.00}
            & 0.00 & 57.75 & \underline{84.34} & \cellcolor{cyan!15} \textbf{100.00} \\
\midrule

\multirow{4}{*}{NQ}
& Badnets  & 0.00 & \underline{95.60} & 0.00 & \cellcolor{cyan!15} \textbf{100.00}
            & 0.00 & \underline{86.83} & 0.00 & \cellcolor{cyan!15} \textbf{98.04} \\
& Addsent  & 0.00 & \underline{94.80} & 0.00 & \cellcolor{cyan!15} \textbf{100.00}
            & 0.00 & \underline{84.30} & 0.00 & \cellcolor{cyan!15} \textbf{94.88} \\
& CBA      & 0.00 & \underline{87.80} & 0.00 & \cellcolor{cyan!15} \textbf{99.40}
            & 0.00 & \underline{72.50} & 0.00 & \cellcolor{cyan!15} \textbf{93.60} \\
& StyleBkd & 0.00 & 82.40 & \underline{88.40} & \cellcolor{cyan!15} \textbf{100.00}
            & 0.00 & 75.33 & \underline{89.11} & \cellcolor{cyan!15} \textbf{96.53} \\
\bottomrule
\end{tabular}
\end{table*}
\subsection{Model performance}
The effectiveness of our proposed \sys against backdoor attacks is evaluated across multiple datasets and models, thereby demonstrating its efficacy and robustness in diverse settings. The results are summarized in \textcolor{cyan}{\autoref{tab:main-cacc}}. 

\textbf{Backdoor Attack Mitigation.} Across all datasets, \sys demonstrates a significant reduction in ASR compared to baseline methods, regardless of whether they employ defensive strategies. When poisoned samples are distributed in an IID manner, \sys consistently achieves an absolute advantage across all datasets and attack methods: \textbf{reducing ASR to zero}. This effectively neutralizes the backdoor threat. Among the baselines evaluated, GraCeFul shows the second-best performance. In contrast, FreqFed, which is highly effective against backdoor attacks originating from malicious clients, exhibits limited effectiveness when confronted with widely distributed untrusted data. This is consistent with the limited number of model updates that contain backdoor features, which illustrates the vulnerability of current FL backdoor defense methods against new threats.\looseness=-1

\textbf{Model Utility Preservation.} Preserving the utility of the model is essential for any defense mechanism. Across all 16 settings, \sys achieves a higher CACC than the baseline without defense in nine settings. Among the nine defense algorithms evaluated, all except ONION are able to maintain the model's CACC at a comparable level. We find that this is because many samples after ONION processing lost their core words. \looseness=-1

\subsection{Poisoned Sample Detection Metrics}
Beyond evaluating ASR, we further investigate detection metrics for sample-wise defense methods against poisoned samples. The results highlight the superior capability of \sys in accurately identifying poisoned samples.\looseness=-1

\textbf{Recall.} Recall measures the proportion of poisoned samples that are mistakenly classified as clean, which directly contributes to the persistence of backdoor mappings during model training. Consequently, achieving a higher recall is critical for effective backdoor defense. As presented in \textcolor{cyan}{\autoref{tab:main-recall-f1}}, \sys consistently demonstrates superior performance across all datasets and attack types, achieving near-perfect or perfect recall in every evaluated setting. This result underscores \sys's ability to reliably identify virtually all backdoor-infected samples, thereby addressing the fundamental cause of backdoor persistence during fine-tuning. On FreebaseQA and CoQA, \sys maintains 100\% recall in all attack variants, while competitors such as ONION and CUBE fail to detect most or all malicious samples (recall = 0.00). Although GraCeFul shows moderate recall in some datasets, its performance degrades significantly under complex attack types such as CBA.

\textbf{F1 Score.} As reported in \textcolor{cyan}{\autoref{tab:main-recall-f1}}, the F1 score further underscores the advantage of \sys: it achieves not only high recall but also strong precision, particularly on FreebaseQA and CoQA, where the F1 score approaches or exceeds 90\%. It is worth noting that under the Badnets attack, even though \sys only achieved an F1 Score of 79.12\% on WebQA (CBA attack), it still maintains a recall rate of 99.08\%, and almost never missed truly poisoned samples even in challenging situations. We notice that the Recall and F1 values of CUBE and ONION are both 0 in most cases, thus we further explore the reasons. For CUBE, comparing our results with those in~\cite{wu2025gracefully}, we find that CUBE is highly dependent on the number of samples. We believe that this is related to two factors. First, CUBE needs to use local data to train an auxiliary model as a feature extractor. Too little data will limit the performance of the feature extractor. Second, CUBE selects the hidden\_state of the last transformer block as the feature to distinguish clean samples from poisoned samples. When the number of samples is insufficient, it lacks discrimination compared to the features in the frequency space.
For ONION, we find that when the number of triggers exceeds one token or when using sentences as triggers like Addsent, ONION can only eliminate one token or even none of the trigger words. More seriously, ONION is prone to collapse and destroy the entire sentence, which is why CACC is significantly reduced when using ONION defense. When faced with StyleBkd, although ONION processed many poisoned samples, this processing does not really change the Bible style of the sentences.

These results confirm that \sys is a highly effective and reliable sample-wise defense method, capable of identifying backdoor-infected inputs in a wide variety of data sets and attack strategies, thus significantly enhancing the robustness of the model in realistic fine-tuning pipelines.

\subsection{Further Analysis}
\subsubsection{Applicability to different models}
Considering the applicability of the defense method to different LLMs, we select widely used representative LLMs (with more than 9 million downloads on \texttt{huggingface} in total) ranging from 1B to 70B in size for experiments. As shown in \textcolor{cyan}{\autoref{tab:model}}, due to gaps in model capabilities, CACC varies between different models, but \sys can always achieve a high recall rate and completely eliminate ASR, and will not interfere with CACC.\looseness=-1

\begin{table}[]
\centering
\caption{Backdoor defense performance of \sys when adopting to different models on FreebaseQA. Considering that different models have different numbers of transformer layers, we choose \texttt{self\_attn.v\_proj.lora\_B} of the last transformer layer as the target module.}
\label{tab:model}
\resizebox{\columnwidth}{!}{%
\begin{tabular}{@{}c|cc|cccc@{}}
\toprule
                        & \multicolumn{2}{c|}{\textbf{Nodefense}} & \multicolumn{4}{c}{\sys}                                                    \\ \cmidrule(lr){2-3} \cmidrule(lr){4-7}
\multirow{-2}{*}{\textbf{Model (\#Layers)}} & \textbf{CACC}           & \textbf{ASR}           & \textbf{CACC}  & \textbf{ASR}  & \textbf{Recall}                         & \textbf{F1}                             \\ \midrule
Llama-3.2-1B (16)             & 26.65          & 80.10         & 27.15 & 0.00 & 94.80                          & 67.96                        \\

Llama-3.2-3B (28)             & 57.95          & 98.35         & 58.50 & 0.00 & 100.00 & 100.00                         \\
GPT-J-6B (28)                &  23.65         & 96.05         & 23.95 & 0.00 & 100.00                          & 100.00                         \\
Llama-2-7B (32)                  & 64.25          &  99.40        & 64.25 & 0.00 & 100.00                          & 99.90                         \\
Vicuna-7B (32)                  & 62.25          &  78.15        & 64.85 & 0.00 & 100.00                          & 100.00 \\
Mistral-7B (32)                 & 68.25          &  97.80        & 69.10 & 0.00 & 99.80                          & 99.90                              \\
Gemma-2-9B (42)               & 68.05          & 91.30         & 68.65 & 0.00 & 99.60                          & 76.73
              \\
Llama-2-13B (40)               & 72.25          & 99.30         & 71.90 & 0.00 & 100.00  & 100.00 \\
Llama-2-70B (80)               & 73.50         & 97.90         & 82.90 & 0.00 & 100.00 & 100.00 \\
\bottomrule
\end{tabular}%
}
\end{table}

\subsubsection{Data distribution}
\label{sec:distribution}
More challenging data distributions are also explored, where data instances are unevenly distributed among clients, mimicking real-world conditions more closely. Despite the increased difficulty, \sys shows impressive resilience, markedly reducing ASR in complex distributions.\looseness=-1

\begin{table}[]
\centering
\caption{In an ideal case of completely clean data, the CACC of training w/ and w/o \sys.}
\label{tab:clean-data}
\setlength{\tabcolsep}{8.5pt}
\begin{tabular}{ccccc}
\toprule
\textbf{Method}              & \textbf{WebQA} & \textbf{FreebaseQA} & \textbf{CoQA}  & \textbf{NQ}    \\ \midrule
\rowcolor{red!10} w/o \sys               & 45.86 & 64.95           & 73.09      &  74.80 \\
\rowcolor{cyan!15} w/ \sys & 45.28 & 64.05      & 72.29 & 73.95 \\ \bottomrule
\end{tabular}
\vskip -0.1in
\end{table}
\textbf{The case without poisoned data.}
It is important to ensure the bottom line of \sys, that is, when there are no poisoned data in the training data, \sys must maintain sufficient precision to avoid damaging too many clean samples. As shown in \textcolor{cyan}{\autoref{tab:clean-data}}, we compare the ideal clean data with the case without defense, and the use of \sys has a loss of less than 1\% on CACC. This shows that the differences between different clean samples are smaller than the differences between clean samples and poisoned samples, and \sys can reduce misjudgments to some extent. Considering the huge risk of backdoor attacks that \sys can avoid, the cost of this part of security is acceptable.\looseness=-1

\textbf{Poison Ratio.}
\label{sec:pr}
\begin{figure}[t]
  \centering
  \includegraphics[width=\linewidth]{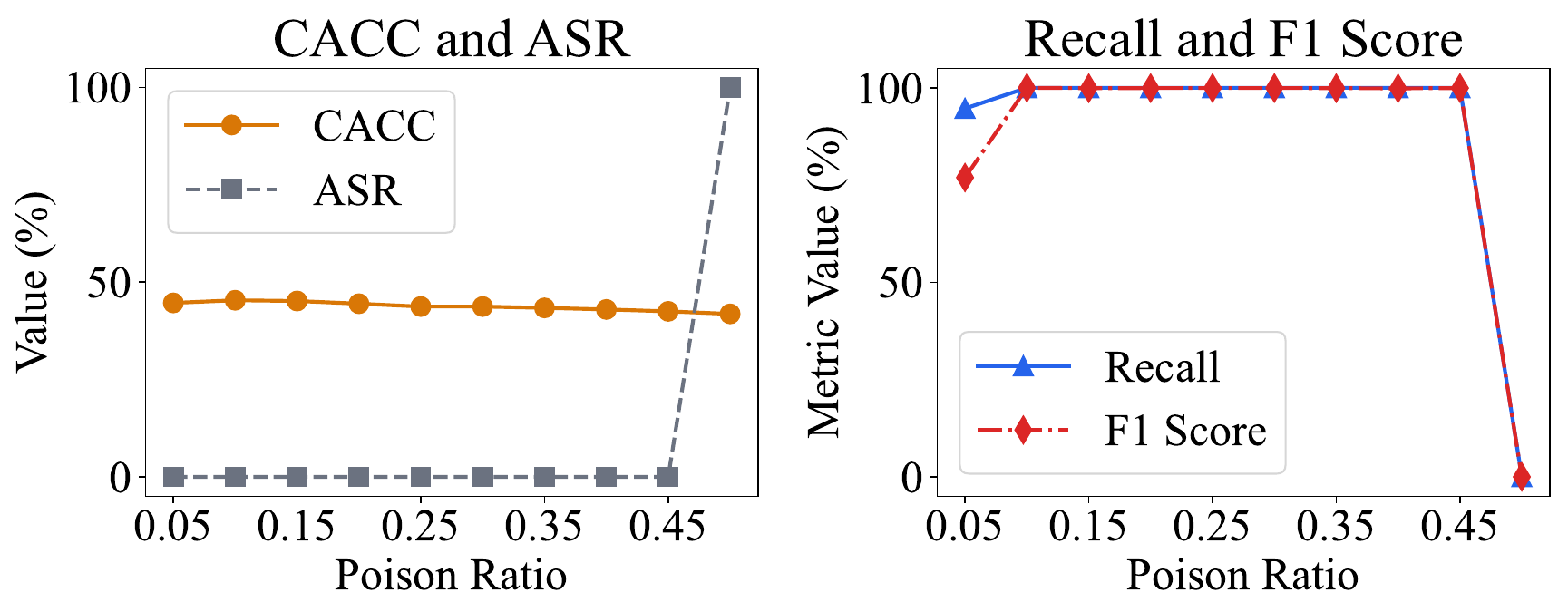}
  \caption{The effectiveness of \sys in defending against different poisoning ratios. Experiments adopt the same IID setting as main results.}
  \label{fig:pr}
  \vskip -0.1in
\end{figure}
To study the generalization of \sys, we first explore the impact of the poison ratio on defense performance in the IID setting. As shown in \textcolor{cyan}{\autoref{fig:pr}}, \sys can always achieve high recall and completely eliminate ASR when clean data is predominant. As the poison ratio gradually increases, \sys filters out more samples, resulting in a slight decrease in CACC. Consistent with the threat model we analyze before, when the data on all clients are mainly poisoned data, the defense effect of \sys will be invalid. In fact, it is difficult to have such a high poison ratio in real scenarios because this will make the attack lose its concealment and have a significant impact on the main task. All of this verifies the excellent generalization of \sys to poison ratios.

\begin{figure*}[!t]
  \centering
  \includegraphics[width=0.99\linewidth]{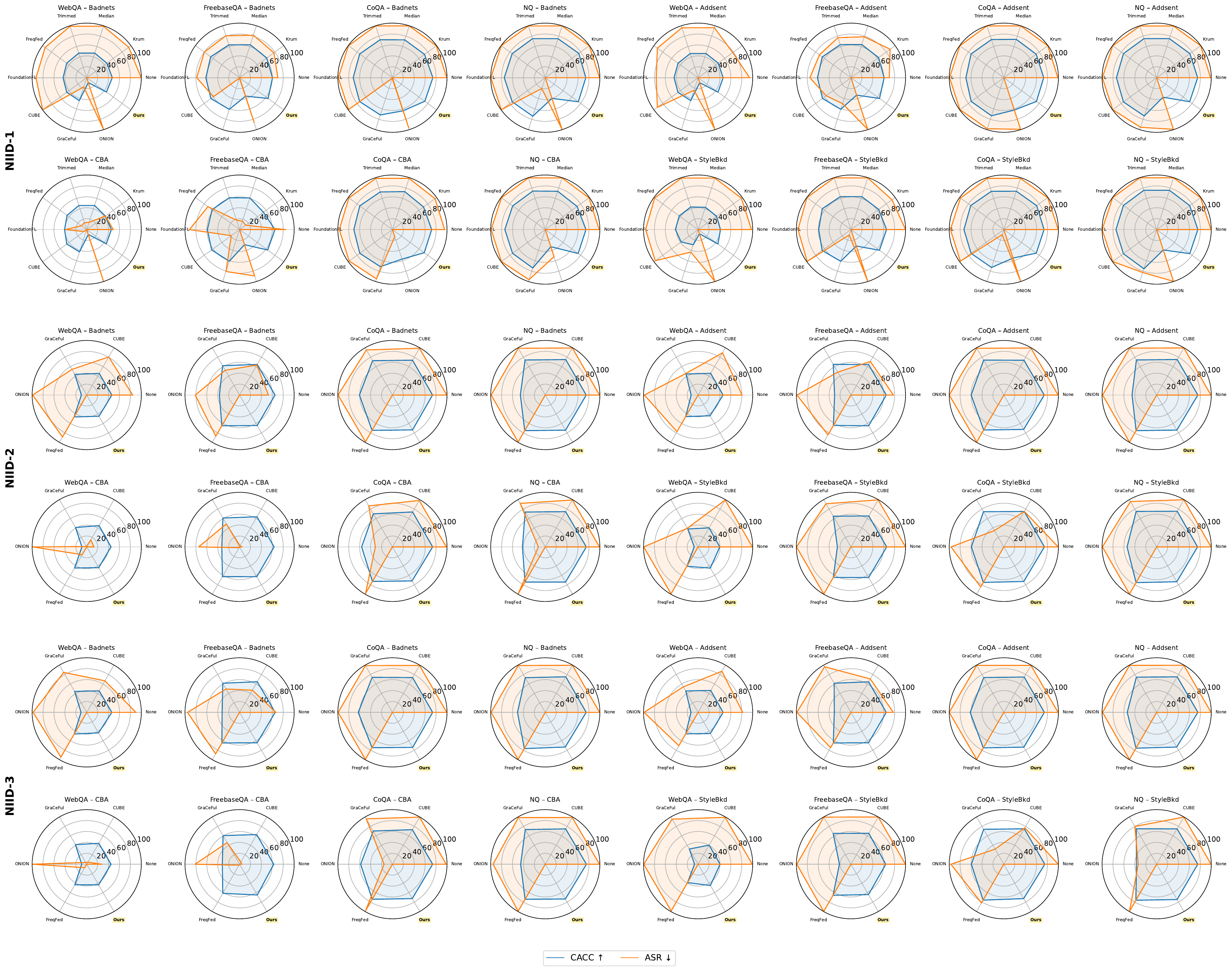}
  \caption{Evaluation of \sys and baseline defense methods on NIID settings.}
  \label{fig:radar}
  \vskip -0.1in
\end{figure*}
\begin{figure*}[!t]
  \centering
  \includegraphics[width=0.99\linewidth]{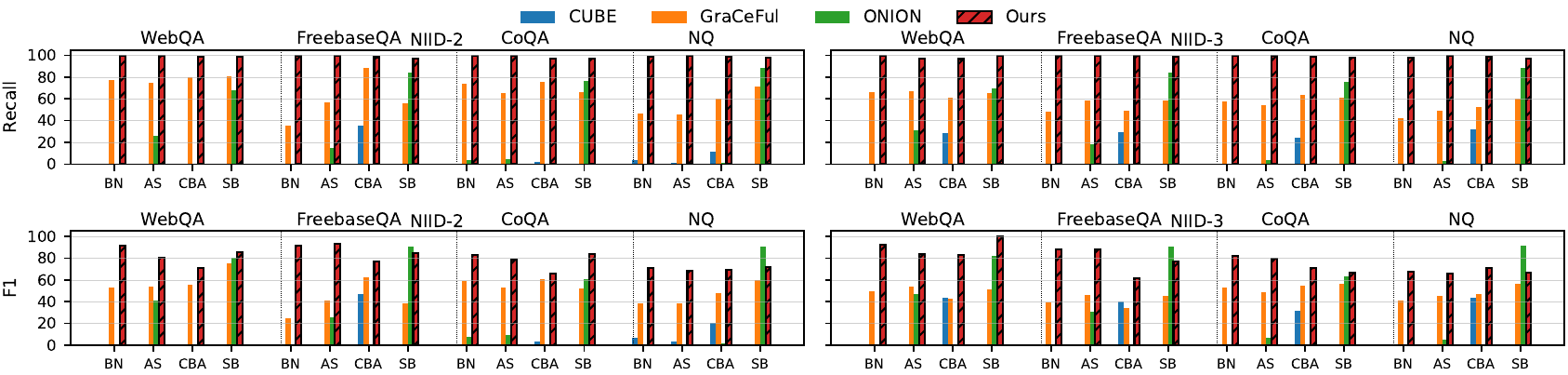}
  \caption{Recall and F1 values of poisoned sample detection by the sample-wise defense method on NIID settings.}
  \label{fig:bar-merge}
  \vskip -0.1in
\end{figure*}

\begin{table*}[!t]
\centering
\caption{Recall and F1 of poisoned sample detection on NIID-1 data. All other settings are the same as \textcolor{cyan}{\autoref{tab:main-cacc}}.}
\label{tab:recall-f1-niid1}
\setlength{\tabcolsep}{8.8pt}
\begin{tabular}{cc|cccc|cccc}
\toprule
\multirow{2}{*}{\textbf{Dataset}}
& \multirow{2}{*}{\textbf{Poison Method}}
& \multicolumn{4}{c|}{\textbf{Recall} $\uparrow$}
& \multicolumn{4}{c}{\textbf{F1} $\uparrow$} \\
\cmidrule(lr){3-6} \cmidrule(lr){7-10}
&
& \textbf{CUBE} & \textbf{GraCeFul} & \textbf{ONION} & \textbf{\sys}
& \textbf{CUBE} & \textbf{GraCeFul} & \textbf{ONION} & \textbf{\sys} \\
\midrule

\multirow{4}{*}{WebQA}
& Badnets  & 0.00 & \underline{86.59} & 0.00 & \cellcolor{cyan!15} \textbf{99.12}
            & 0.00 & \underline{82.66} & 0.00 & \cellcolor{cyan!15} \textbf{99.43} \\
& Addsent  & 0.00 & \underline{51.08} & 36.91 & \cellcolor{cyan!15} \textbf{98.49}
            & 0.00 & 10.46 & \underline{53.92} & \cellcolor{cyan!15} \textbf{84.99} \\
& CBA      & 0.00 & \underline{84.91} & 0.00 & \cellcolor{cyan!15} \textbf{98.29}
            & 0.00 & \underline{65.19} & 0.00 & \cellcolor{cyan!15} \textbf{81.87} \\
& StyleBkd & 0.00 & \underline{81.61} & 67.89 & \cellcolor{cyan!15} \textbf{98.83}
            & 0.00 & 66.90 & \underline{80.80} & \cellcolor{cyan!15} \textbf{88.87} \\
\midrule

\multirow{4}{*}{FreebaseQA}
& Badnets  & 0.00 & \underline{96.34} & 0.00 & \cellcolor{cyan!15} \textbf{100.00}
            & 0.00 & \underline{83.06} & 0.00 & \cellcolor{cyan!15} \textbf{94.99} \\
& Addsent  & 0.00 & \underline{69.79} & 21.02 & \cellcolor{cyan!15} \textbf{99.75}
            & 0.00 & 21.72 & \underline{34.74} & \cellcolor{cyan!15} \textbf{92.71} \\
& CBA      & 7.31 & \underline{56.28} & 0.10 & \cellcolor{cyan!15} \textbf{99.90}
            & 12.79 & \underline{21.76} & 0.10 & \cellcolor{cyan!15} \textbf{95.15} \\
& StyleBkd & 0.00 & \underline{97.14} & 86.44 & \cellcolor{cyan!15} \textbf{99.63}
            & 0.00 & 82.38 & \textbf{92.05} & \cellcolor{cyan!15} \underline{89.05} \\
\midrule

\multirow{4}{*}{CoQA}
& Badnets  & 0.00 & \underline{96.82} & 0.00 & \cellcolor{cyan!15} \textbf{98.98}
            & 0.00 & \underline{85.72} & 0.00 & \cellcolor{cyan!15} \textbf{92.52} \\
& Addsent  & 0.00 & \underline{38.96} & 0.00 & \cellcolor{cyan!15} \textbf{100.00}
            & 0.00 & \underline{18.27} & 0.00 & \cellcolor{cyan!15} \textbf{100.00} \\
& CBA      & 0.00 & \underline{39.25} & 0.00 & \cellcolor{cyan!15} \textbf{100.00}
            & 0.00 & \underline{12.70} & 0.00 & \cellcolor{cyan!15} \textbf{94.93} \\
& StyleBkd & 0.00 & \underline{95.15} & 71.58 & \cellcolor{cyan!15} \textbf{100.00}
            & 0.00 & \underline{90.41} & 61.63 & \cellcolor{cyan!15} \textbf{100.00} \\
\midrule

\multirow{4}{*}{NQ}
& Badnets  & 0.00 & \underline{89.88} & 0.00 & \cellcolor{cyan!15} \textbf{99.25}
            & 0.00 & \underline{81.17} & 0.00 & \cellcolor{cyan!15} \textbf{89.82} \\
& Addsent  & 0.00 & \underline{47.92} & 0.00 & \cellcolor{cyan!15} \textbf{99.75}
            & 0.00 & \underline{16.36} & 0.00 & \cellcolor{cyan!15} \textbf{83.59} \\
& CBA      & 0.00 & \underline{34.20} & 0.00 & \cellcolor{cyan!15} \textbf{99.11}
            & 0.00 & \underline{14.92} & 0.00 & \cellcolor{cyan!15} \textbf{94.89} \\
& StyleBkd & 0.00 & 78.11 & \underline{91.04} & \cellcolor{cyan!15} \textbf{99.75}
            & 0.00 & 71.08 & \textbf{92.25} & \cellcolor{cyan!15} \underline{82.51} \\
\bottomrule
\end{tabular}
\end{table*}
\textbf{NIID-1 setting.} Based on IID, NIID-1 takes into account that different clients use different sources of training data, which may lead to different data poisoning ratios.
Specifically, the poison ratio of each client is randomly sampled between [0, 0.4]. As shown in \textcolor{cyan}{\autoref{fig:radar}} and \textcolor{cyan}{\autoref{tab:recall-f1-niid1}}, \sys continues to achieve near-perfect performance on ASR and Recall (ASR = 0, Recall $\rightarrow$ 100.00\%), and the high F1 score ensures that this process does not have a significant negative impact on CACC.
At the same time, there is also a large gap between the second-best defense method and \sys in defending poisoned samples.

\textcolor{cyan}{\autoref{fig:radar}} and \textcolor{cyan}{\autoref{fig:bar-merge}} show the results for more complex settings. \textbf{NIID-2:} Poison rations are samples based on a more realistic and unbalanced Dirichlet distribution, with $\alpha = 0.5$. At the same time, NIID-2 limits the maximum local poison ratio of each client to 0.5, ensuring that clean data dominate local data, which is also a condition for many defense solutions to apply. Results show more complex scenarios do not pose a great challenge to the performance of \sys. ASR is stably eliminated to 0, and the recall in the worst case was still 97.04\%, far ahead of other baselines.
\textbf{NIID-3:} Finally, we consider the most challenging scenario. NIID-3 removes the upper limit of the 0.5 poison ratio based on NIID-2 to simulate that some clients are caught with the worst data sampling. Results show the results. In this scenario, the second best method can only achieve a recall of about 50\%, which means that about half of the poisoned data are missed by the defense algorithm, posing a huge risk to model training. \sys maintained a balance between reducing ASR and preserving CACC, even with skewed data distributions. It can correct client-level misclassifications.

\textcolor{cyan}{\autoref{tab:alpha}} further evaluate the detection performance of \sys with different $\alpha$ in $\{0.1,0.3,0.7,1.0\}$, results show ProtegoFed achieves near-perfect Recall ($>$96\%).

\begin{table}[!t]
\centering
\caption{Recall on various $\alpha$ in NIID FreebaseQA.}
\label{tab:alpha}
\footnotesize
\setlength{\tabcolsep}{8.5pt}
\begin{tabular}{c|cc|cc}
\toprule
\multirow{2}{*}{$\alpha$} &
\multicolumn{2}{c|}{\textbf{NIID-2}} &
\multicolumn{2}{c}{\textbf{NIID-3}}\\
\cmidrule(lr){2-3} \cmidrule(lr){4-5}
& GraCeFul & \sys & GraCeFul & \sys \\
\midrule
0.1 & 38.18 & \cellcolor{cyan!15} 97.72 & 52.55 & \cellcolor{cyan!15} 99.56 \\
0.3 & 47.80 & \cellcolor{cyan!15} 96.38 & 55.30 & \cellcolor{cyan!15} 96.86 \\
0.7 & 52.36 & \cellcolor{cyan!15} 98.64 & 64.87 & \cellcolor{cyan!15} 98.81 \\
1.0 & 71.86 & \cellcolor{cyan!15} 99.00 & 61.69 & \cellcolor{cyan!15} 99.34 \\
\bottomrule
\end{tabular}
\vskip -0.1in
\end{table}

\begin{table}[]
\centering
\caption{The impact of global secondary clustering and client correction on the detection performance of poisoned samples when most clients are dominated by poisoned data.}
\label{tab:dominate}
\setlength{\tabcolsep}{11pt}
\begin{tabular}{ccccc}
\toprule
\textbf{Method}     & \textbf{Accuracy} & \textbf{Precision} & \textbf{Recall} & \textbf{F1}    \\ \midrule
\rowcolor{red!10} w/o revise & 45.90    & 22.68     & 11.09  & 14.89 \\
\rowcolor{cyan!15} w/ revise  & 99.29    & 99.58     & 98.76  & 99.17 \\ \bottomrule
\end{tabular}%
\vskip -0.1in
\end{table}
\textbf{The case where poisoned data dominates.} The starting point of \sys design is that clean data can dominate on most clients, so that the global centroid feature remains on the side of clean samples to guide clients to perform local data correction. However, based on the results in \textcolor{cyan}{\autoref{fig:pr}}, we further study another situation, that is, more than half of the clients are dominated by poisoned data, but these poisoned data come from different attacking parties, with different triggers and target outputs. We find that in this case, through carefully designed global secondary clustering, the global center can still fall into the clean data area. Therefore, we design a simple scenario for verification: On 40\% of the clients, clean data dominates, with only 10\% poisoned data; the remaining clients are divided into two groups, each with 30\% of the clients, and their local poisoned data dominate (70\% poisoning rate). The results are shown in \textcolor{cyan}{\autoref{tab:dominate}}. When a relatively large number of clients can provide accurate local centers, \sys can still achieve accurate corrections globally through global secondary clustering.

\textbf{Comparison Across Distributions.} The performance of \sys is stable between different distributions compared to baselines. The method's adaptability is attributed to its strategic use of frequency-based global information exchange, which enhances its capability to distinguish backdoor patterns irrespective of data distribution bias. In summary, \sys's robust architecture allows it to perform consistently across a spectrum of data distributions, ensuring its effectiveness for practical federated learning applications, where data distribution can vary significantly.

\subsubsection{Target Module Selection}
\begin{table}[t]
\centering
\caption{The impact of target module selection. All experiments are conducted on the WebQA. \texttt{layers.31.lora\_B} denotes the weights of \texttt{model.layers.31.self\_attn.v\_proj.lora\_B}}
\label{tab:module}
\resizebox{\columnwidth}{!}{%
\begin{tabular}{cccccc}
\toprule
\textbf{Module}          & \textbf{CACC}  & \textbf{ASR}   & \textbf{Recall} & \textbf{F1}     & \textbf{\# Param}  \\ \midrule
\rowcolor{cyan!15} \textbf{\texttt{layers.31.lora\_B}} & \textbf{45.37} & \textbf{0.00}  & \textbf{100.00} & \textbf{100.00} & \textbf{16,384}     \\
\texttt{layers.31.lora\_A} & 43.60 & 32.38 & 36.92  & 13.79  & 16,384     \\
\texttt{layers.15.lora\_B} & 44.49 & 0.00  & 82.15  & 58.17  & 16,384     \\
\texttt{layers.0.lora\_B} & 43.70 & 58.42 & 28.00  & 15.14  & 16,384     \\
\texttt{embed\_tokens}   & 43.55      & 5.95      & 53.85       & 27.28       & 131,072,000 \\
\texttt{lm\_head}        & 44.54 & 0.00  & 76.31  & 59.05  & 131,072,000 \\ \bottomrule
\end{tabular}%
}
\vskip -0.1in
\end{table}
Determining the optimal module for the calculation of the gradient, which is considered the feature of the sample, is crucial to the efficacy of \sys. Our experiments focus on different parts of LLMs to understand their impact on the detection of poisoned samples. Specifically, we explored the following options: (i) Inside the transformer block, \texttt{lora\_B} modules at layers 0, 15, and 31 (32 layers in total) and \texttt{lora\_A} module; (ii) Outside the transformer block, the \texttt{embed} layer close to the input and the \texttt{lm\_head} layer closet to the output.\looseness=-1

As shown in \textcolor{cyan}{\autoref{tab:module}}, the best results are achieved when using \texttt{lora\_B} in \texttt{layers.31} as the target parameters, which is also the setup we use in the main experiments. When using the module in the \texttt{transformer block}, the discrimination of the poisoned samples gradually increases as the number of layers increases from shallow to deep ($0\rightarrow31$). This is consistent with the theoretical basis of our approach in Section~\textcolor{cyan}{\ref{sec:freq}} that deeper layers will amplify the differences in frequency features, making the poisoned samples more discriminable. When \texttt{lora\_A} and \texttt{lora\_B} of the same layer are used as target modules, \texttt{lora\_B} has a significant performance advantage over \texttt{lora\_A}. We believe that this is due to the different initialization of \texttt{lora\_A} and \texttt{lora\_B}, where \texttt{lora\_A} parameters are initialized by Gaussian distribution, while \texttt{lora\_B} is initialized with all zeros. Then when the model training starts to calculate the sample-wise gradient, the parameters of \texttt{lora\_A} will easily get an excessively large offset value, causing deviations in the clustering algorithm; while \texttt{lora\_B} initialized with all zeros provides a standard and unified starting point for all parameters, which facilitates the use of cosine similarity for reduction and clustering in high-dimensional space.

When the linear layer at the input (\texttt{embed\_token}) and output (\texttt{lm\_head}) of the model is used as the target module, the time cost of directly using UMAP to reduce the dimension is huge due to the huge increase in parameter dimensions compared to the LoRA module ($16,384\rightarrow131,072,000$). Thus for \texttt{embed\_token} and \texttt{lm\_head}, we adopt the same method as~\cite{wu2025gracefully} that first apply PCA to reduce dimension. Compared with the input layer, the output layer has a significantly better discrimination for poisoned samples, which is consistent with the conclusion of the above analysis. At the same time, the results obtained by both are not as good as the final LoRA module. We found that this is due to the limited amount of data for each client in FL, which cannot make a clear distinction between particularly high-dimensional features (dimensionality curse).

\begin{table*}[t]
\centering
\caption{The impact of using silhouette score as the selection criterion of clustering method.}
\label{tab:silhouette}
\setlength{\tabcolsep}{4.8pt}
\begin{tabular}{cc|ccccc|ccccc}
\toprule
\multirow{2}{*}{\textbf{Metric}} & \multirow{2}{*}{\textbf{Method}} & \multicolumn{5}{c|}{\textbf{NIID-2}}         & \multicolumn{5}{c}{\textbf{NIID-3}}           \\ \cmidrule(lr){3-7} \cmidrule(lr){8-12}
                        &                         & \textbf{WebQA}  & \textbf{FreebaseQA} & \textbf{CoQA}  & \textbf{NQ} &
                        \textbf{Average}   & \textbf{WebQA}  & \textbf{FreebaseQA} & \textbf{CoQA}   & \textbf{NQ} & \textbf{Average}   \\ \midrule
\multirow{4}{*}{Recall} & w/o HDBSCAN             & \cellcolor{red!10} 99.68  & \cellcolor{red!10} 99.64      & \cellcolor{red!10} 98.58 & \cellcolor{red!10} 98.46 & \cellcolor{red!10} 99.09(0.44$\downarrow$) & \cellcolor{red!10} 99.72  & \cellcolor{red!10} 99.60      & \cellcolor{red!10} 100.00 & \cellcolor{red!10} 98.27 & \cellcolor{red!10} 99.40(0.02$\downarrow$) \\
                        & w/o Hier.               & \cellcolor{red!10} 96.31  & \cellcolor{red!10} 96.21      & \cellcolor{red!10} 98.46 & \cellcolor{red!10} 98.70 & \cellcolor{red!10} 97.42(2.11$\downarrow$) & \cellcolor{red!10} 95.60  & \cellcolor{red!10} 90.18      & \cellcolor{red!10} 95.44  & \cellcolor{red!10} 94.94 & \cellcolor{red!10} 94.04(5.38$\downarrow$) \\
                        & w/o Threshold           & \cellcolor{red!10} 98.88  & \cellcolor{red!10} 98.93      & \cellcolor{red!10} 97.39 & \cellcolor{red!10} 98.70 & \cellcolor{red!10} 98.48(1.05$\downarrow$) & \cellcolor{red!10} 98.86  & \cellcolor{red!10} 98.79      & \cellcolor{red!10} 97.98  & \cellcolor{red!10} 97.87 & \cellcolor{red!10} 98.38(1.04$\downarrow$) \\
                        & \sys                    & \cellcolor{cyan!15} 100.00 & \cellcolor{cyan!15} 99.53      & \cellcolor{cyan!15} 99.76 & \cellcolor{cyan!15} 98.82 & \cellcolor{cyan!15} 99.53(0.00)             & \cellcolor{cyan!15} 100.00 & \cellcolor{cyan!15} 99.60      & \cellcolor{cyan!15} 99.80  & \cellcolor{cyan!15} 98.27 & \cellcolor{cyan!15} 99.42(0.00) \\ \midrule
\multirow{4}{*}{F1}     & w/o HDBSCAN             & \cellcolor{red!10} 82.09  & \cellcolor{red!10} 80.10      & \cellcolor{red!10} 76.05 & \cellcolor{red!10} 66.11 & \cellcolor{red!10} 76.09(8.29$\downarrow$)  & \cellcolor{red!10} 83.77  & \cellcolor{red!10} 84.90      & \cellcolor{red!10} 70.85  & \cellcolor{red!10} 67.37 & \cellcolor{red!10} 76.72(5.81$\downarrow$) \\
                        & w/o Hier.               & \cellcolor{red!10} 60.09  & \cellcolor{red!10} 57.20      & \cellcolor{red!10} 60.57 & \cellcolor{red!10} 55.13 & \cellcolor{red!10} 58.25(26.13$\downarrow$) & \cellcolor{red!10} 62.58  & \cellcolor{red!10} 58.20      & \cellcolor{red!10} 62.87  & \cellcolor{red!10} 57.72 & \cellcolor{red!10} 60.34(22.19$\downarrow$) \\
                        & w/o Threshold           & \cellcolor{red!10} 69.76  & \cellcolor{red!10} 73.47      & \cellcolor{red!10} 66.77 & \cellcolor{red!10} 68.48 & \cellcolor{red!10} 69.62(14.76$\downarrow$) & \cellcolor{red!10} 72.16  & \cellcolor{red!10} 74.53      & \cellcolor{red!10} 71.78  & \cellcolor{red!10} 70.01 & \cellcolor{red!10} 72.12(10.41$\downarrow$) \\
                        & \sys                    & \cellcolor{cyan!15} 91.35  & \cellcolor{cyan!15} 91.90      & \cellcolor{cyan!15} 83.00 & \cellcolor{cyan!15} 71.28 & \cellcolor{cyan!15} 84.38(0.00)              & \cellcolor{cyan!15} 92.27  & \cellcolor{cyan!15} 88.21      & \cellcolor{cyan!15} 81.76  & \cellcolor{cyan!15} 67.87 & \cellcolor{cyan!15} 82.53(0.00) \\ \bottomrule
\end{tabular}%
\vskip -0.1in
\end{table*}
Experimental results show that \texttt{lora\_B} in the last transformer block is the best choice.

\subsubsection{The Impact of Silhouette Score}
In \sys, the silhouette score is used as a criterion in two steps: to determine whether to use hierarchical clustering or HDBSCAN, and to determine the rationality of the clustering results. In order to study the rationality of the use of the silhouette score, we remove each mechanism separately for experiments. As shown in \textcolor{cyan}{\autoref{tab:silhouette}}, a high recall can be achieved regardless of whether the silhouette score is used or not. However, neither only use HDBSCAN nor hierarchical clustering can achieve the best F1 score. This is because the stability of a single clustering algorithm is limited and is easily affected by different data distributions. Combining with the silhouette score can significantly reduce the proportion of errors.
In clients without poisoned samples, the silhouette score provides an additional layer of confidence in the clustering results. With clean data, low silhouette scores due to insufficient discrimination consistently reiterate the effectiveness of using frequency domain gradients for clustering purposes, making it a reliable measure to ensure that cases with potential misclassifications are minimized.

The results further reveal that the use of the silhouette score not only strengthens the internal evaluation of the model's clustering effect, but also proves instrumental in adjusting parameters dynamically. This allows for a custom calibration of defense settings that can automatically pivot to prioritizing utility in the absence of any detected threat. Thus, \sys demonstrates a degree of scalability and accuracy, regardless of the presence of poisoned data. \looseness=-1

\subsubsection{The Role of Global Secondary Clustering and Local Revising}
\begin{figure}[t]
  \centering
  \includegraphics[width=\linewidth]{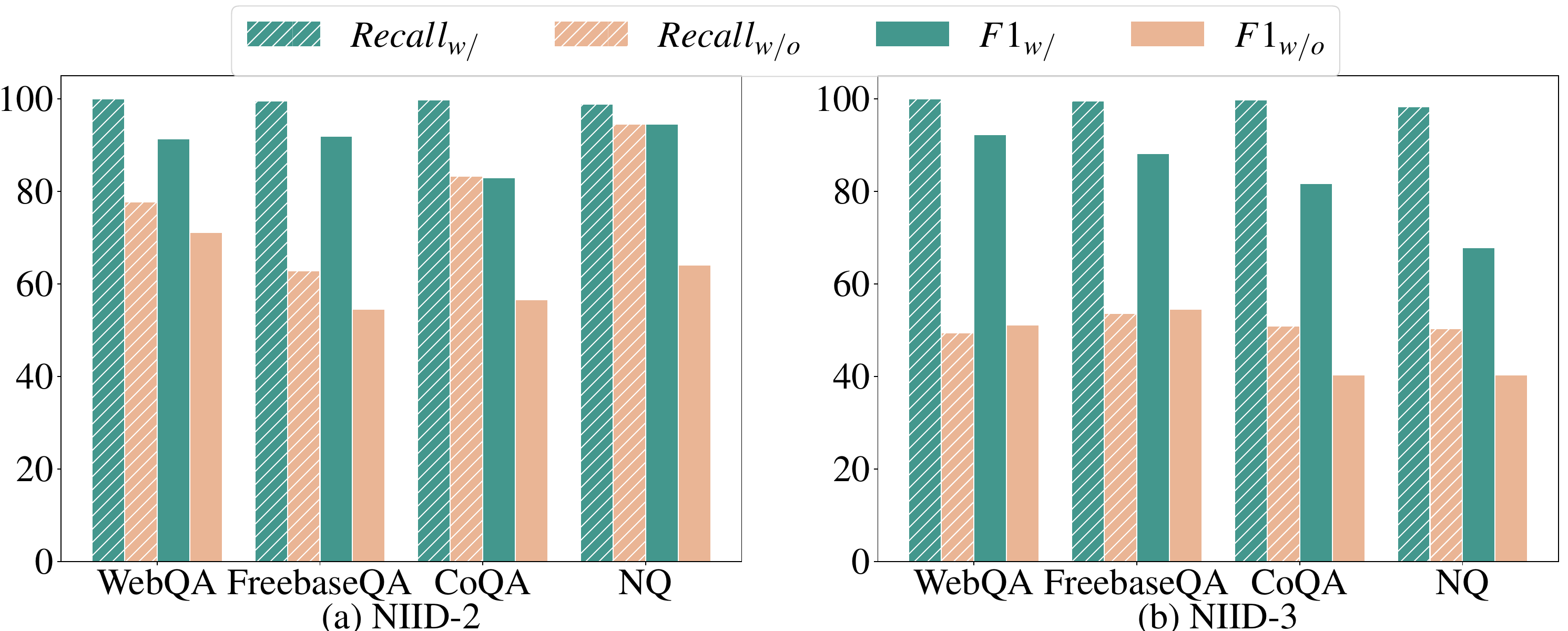}
  \caption{The impact of global secondary clustering then local revising on poisoned sample detection. \textit{w/} and \textit{w/o} denote with and without global re-clustering then local revising, respectively.\looseness=-1}
  \label{fig:revise}
  \vskip -0.1in
\end{figure}
We compare the defense performance of \sys w/ and w/o global secondary clustering and local revising on NIID-2 and NIID-3 settings, the results are shown in \textcolor{cyan}{\autoref{fig:revise}}. Firstly, regardless of the data distribution, the method using revise has better Recall and F1 than the method without revise, which proves that the global information obtained by local centroids has a universal auxiliary effect. Secondly, we can see that the results on NIID-3 show greater differences compared to NIID-2. On NIID-3, the recall without correction is only about 50\%, which is much lower than the result with correction (close to 100\%). This fully demonstrates the significant advantages of \sys in complex and high-risk scenarios. Even if some clients have a large amount of poisoned data (more than 50\%) due to poor sampling and other reasons, \sys can still guide the clients to make corrections through the global center. We also give an example of the detection of samples on the client in this case. \textcolor{cyan}{\autoref{fig:revising}} gives an example of local samples prediction results before and after global secondary clustering and local revising.
Since the number of local poisoned samples exceeds half, the initially obtained local centroid falls in the poisoned area, making the main cluster of the initial local clustering the poisoned samples. The global centroid obtained after global secondary clustering returns to the clean sample area, and the prediction result after local correction based on this is correct.
This is a strong performance and exceeds the general assumptions of previous defense work.

\begin{figure}[t]
    \begin{subfigure}{0.325\linewidth}
      \centering
      \includegraphics[width=\linewidth]{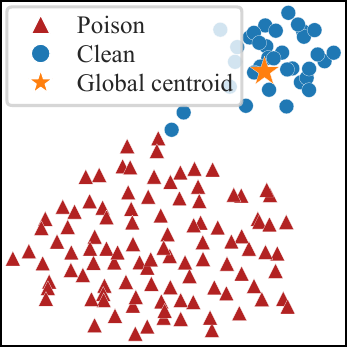}
      \caption{Ground Truth}
    \end{subfigure}
    \hfill
    \begin{subfigure}{0.325\linewidth}
      \centering
      \includegraphics[width=\linewidth]{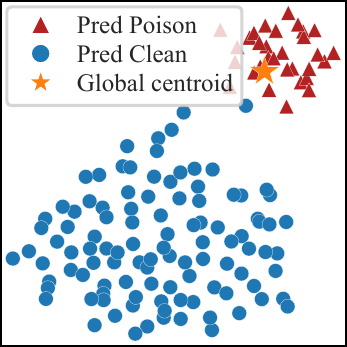}
      \caption{Before Revising}
    \end{subfigure}
    \hfill
    \begin{subfigure}{0.325\linewidth}
      \centering
      \includegraphics[width=\linewidth]{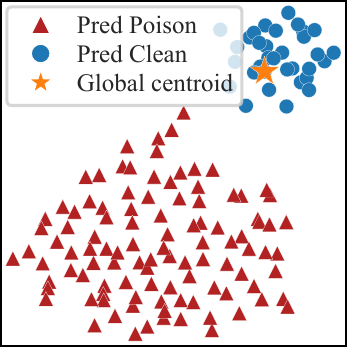}
      \caption{After Revising}
    \end{subfigure}
    \caption{Example of local samples prediction results before and after global secondary clustering and local revising for clients with a high proportion of local poisoned data.}
    \label{fig:revising}
    \vskip -0.1in
  \end{figure}

\subsubsection{The impact of dimensionality reduction}
\label{sec:dimensionality}
We compare the results on WebQA to assess the impact of dimensionality reduction techniques. As shown in \textcolor{cyan}{\autoref{fig:dimension}}, it is difficult to achieve good results by directly clustering a small amount of data in high dimensions.

In addition, we study the effect of reducing the dimensions to different dimensions using UMAP on the recognition of poisoned samples. It can be seen in \textcolor{cyan}{\autoref{fig:umapdimension}} that at lower dimensions, the choice of dimension has little effect on the identification of poisoned samples. For the convenience of visualization, we use 2 dimensions in design of \sys.

\begin{figure}[t]
    \begin{subfigure}{\linewidth}
      \centering
      \includegraphics[width=0.85\linewidth]{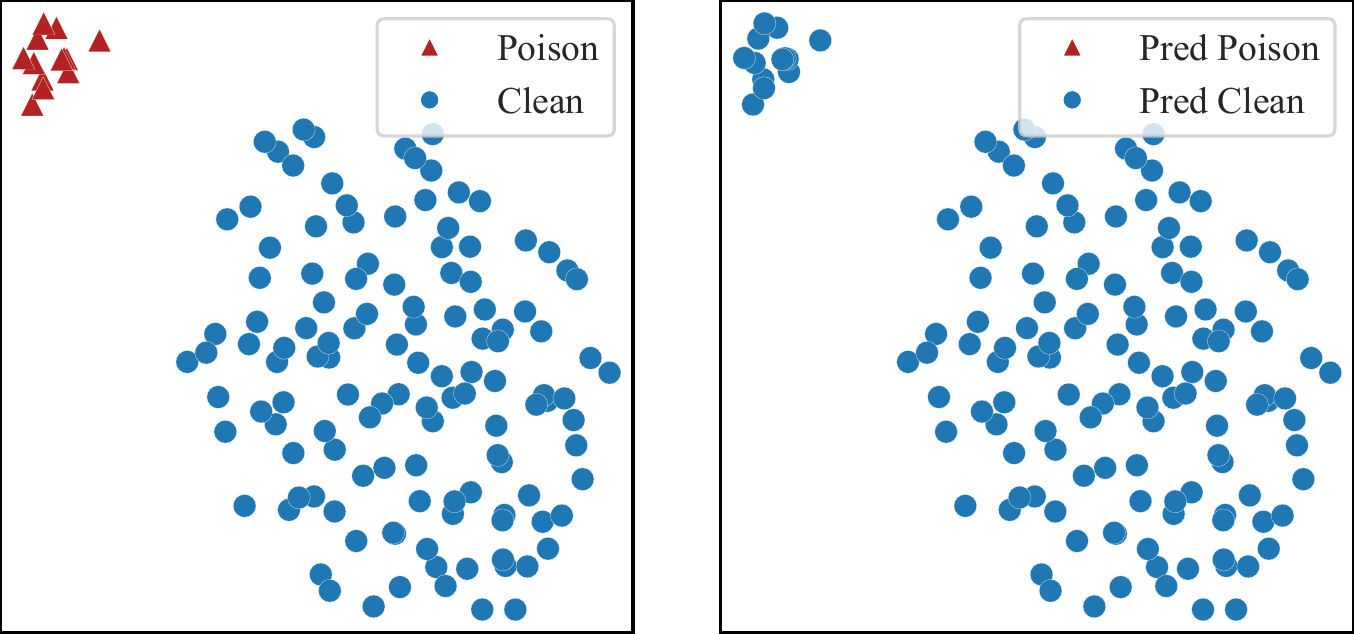}
      \caption{Client 0: High-Dimensional (16384D) Clustering Results}
    \end{subfigure}
    \vfill
    \begin{subfigure}{\linewidth}
      \centering
      \includegraphics[width=0.85\linewidth]{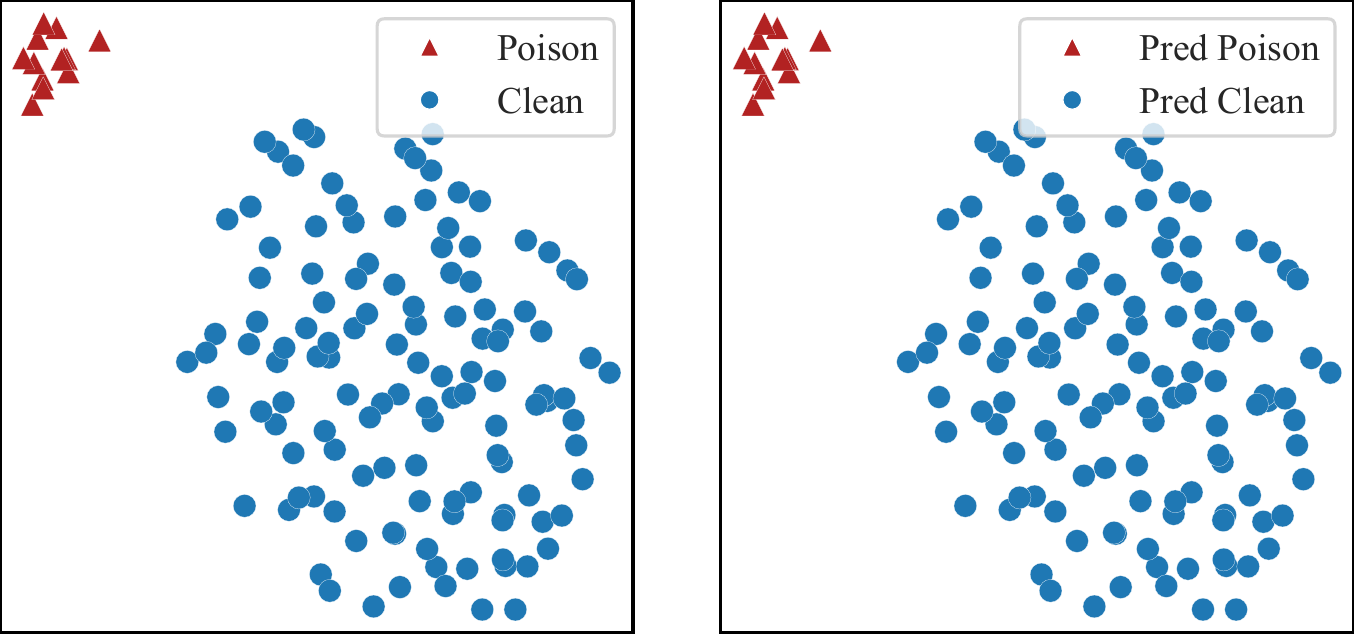}
      \caption{Client 0: Low-Dimensional (2D) Clustering Results}
    \end{subfigure}

    \caption{Visualization of the effect of dimension reduction on the clustering effect. In each group of sub-graphs, the left side is the ground truth, and the right side is the division of clusters obtained by executing the clustering algorithm.}
    \label{fig:dimension}
    \vskip -0.1in
  \end{figure}

\subsubsection{Complexity and client-side memory usage.}
The complexity of \sys is low as its core components, UMAP and HDBSCAN, scale efficiently with the number of samples ($N$). The complexity of UMAP and HDBSCAN is approximately $O(N\log N)$, which is computationally inexpensive for distributed clients. 

In terms of client \textbf{GPU memory} usage, the gradient calculation performed by the client is part of the client's normal training process, and its peak memory usage will not exceed the normal process. We will add a part to further demonstrate its practicality.

\section{Discussion}
\label{discussion}

\subsection{Time Efficiency}

We evaluate the time efficiency of \sys by comparing its runtime overhead with baseline (FedAvg). As shown in \textcolor{cyan}{\autoref{fig:time}}, \sys incurs a small initial overhead before training begins: 12.43 seconds (s) for Intra-Client Clustering, 0.07s for Global Clustering, and 0.39s for Local Revising. Despite this, the overall training time for \sys is shorter than baseline for 100 communication rounds. This improvement is attributed to the removal of poisoned data, which reduces the training size and thus accelerates each training round. Importantly, the time overhead introduced by \sys is one-time and negligible in relation to the total training time (only around 1\%).

\begin{figure}[t]
  \centering
  \includegraphics[width=\columnwidth]{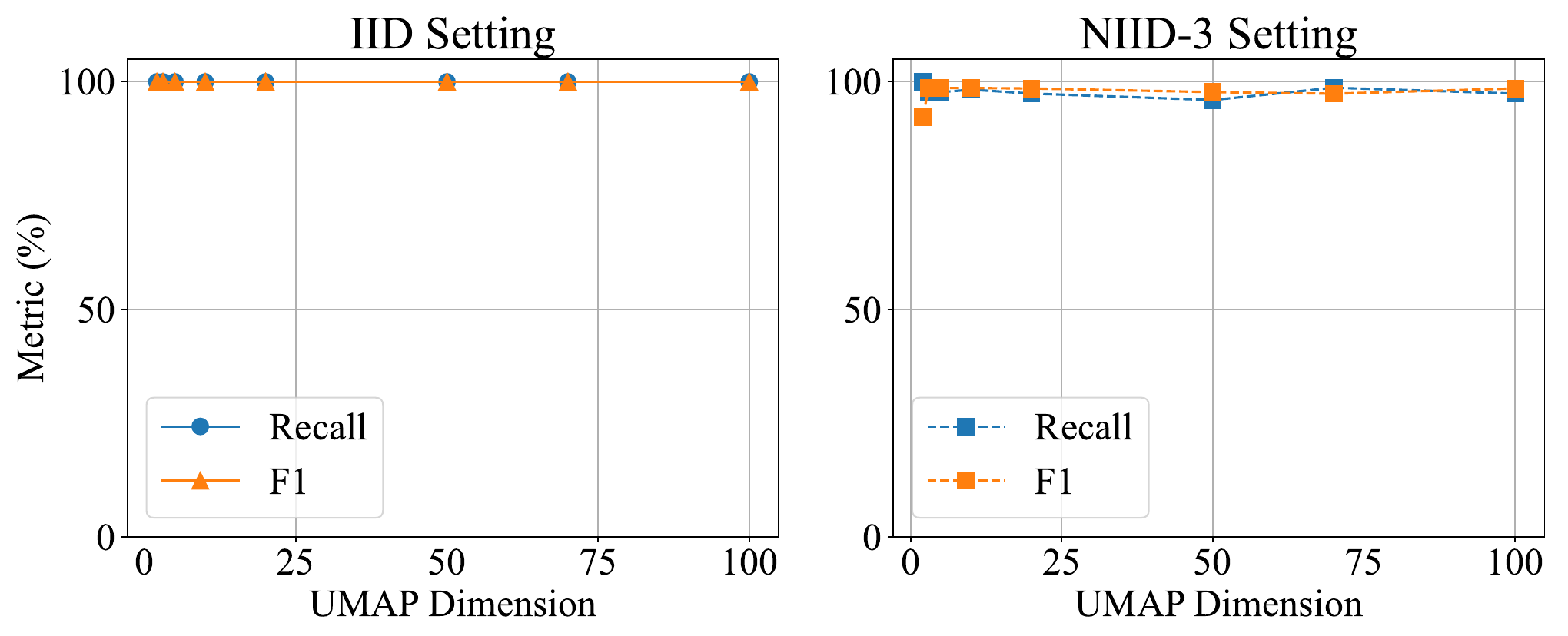}
  \caption{The impact of using UMAP to reduce to different dimensions on the identification of poison samples.}
  \label{fig:umapdimension}
\end{figure}
\begin{figure}[t]
  \centering
  \includegraphics[width=0.8\linewidth]{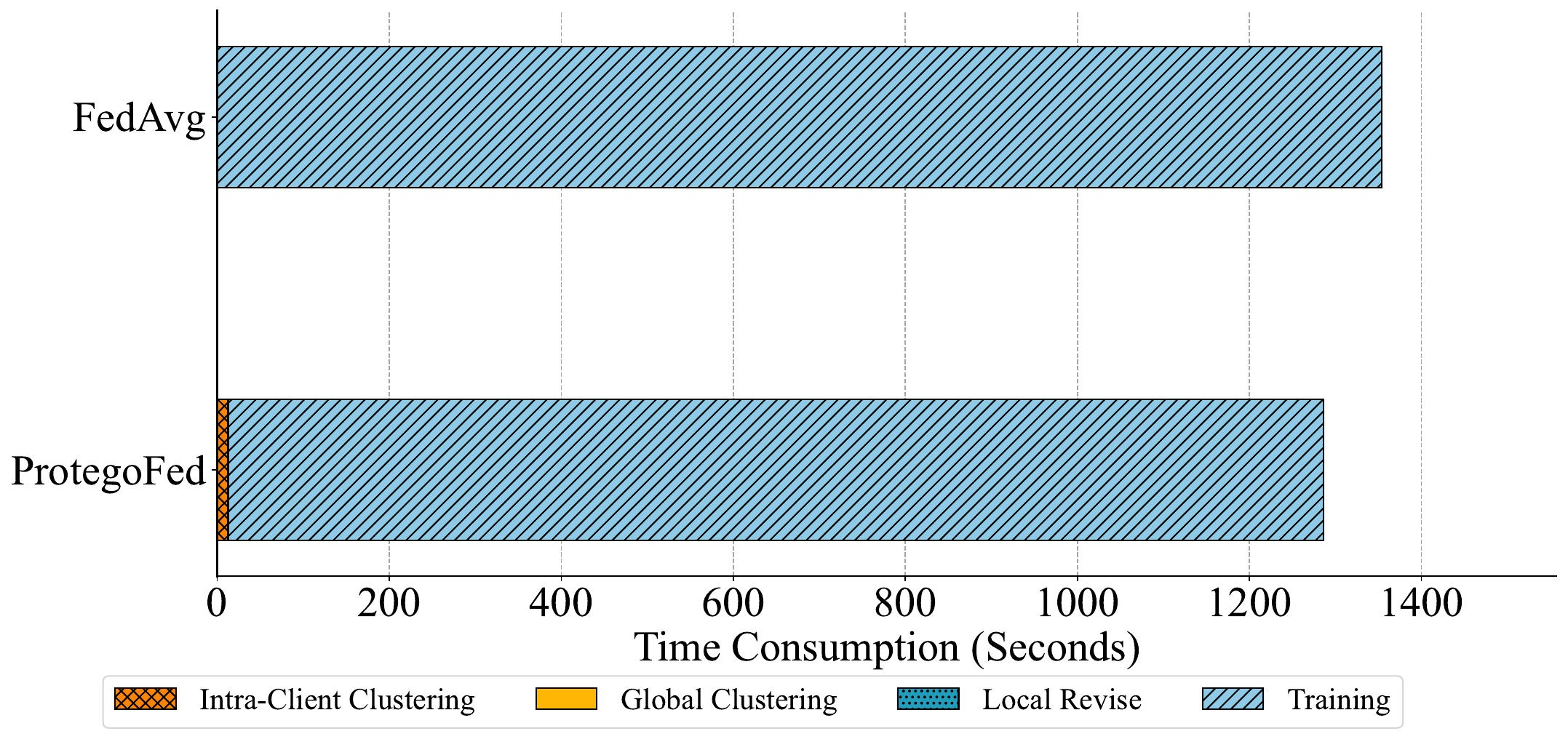}
  \caption{Time overhead of \sys and FedAvg. Results are calculated on the WebQA for 100 rounds, excluding validation and test time.}
  \label{fig:time}
\end{figure}
\subsection{Scalability of \sys}
In real-world FL environments, clients often dynamically join or withdraw from the training process due to connectivity issues, availability, or resource constraints. We evaluate the scalability and robustness of \sys and find it handles dynamic clients. On the one hand, client withdrawal will not affect the one-time filtering that has been performed. On the other hand, as new clients join, Step \ding{182} is uploaded and accumulated on the server. We evaluate the addition based on the existing 25 clients as above. Late-joiners can use established global-centroid and run Step \ding{183} (0.07s)-\ding{185} (0.39s) once more to update individually. Moreover, the global center update brought by more local gradients makes the later-joined clients more robust in filtering.
\subsection{Adaptive Attack Strategy}

Considering that adversaries knowing the defense algorithm may adjust their attack strategies, we further evaluate the robustness of \sys against adaptive attacks. 

\textbf{Frequency-based attacks.} Considering the threat model, the adversaries can only indirectly backdoor model training by modifying the data. As the core method of \sys is to identify backdoor models by using the discrimination of different samples in the model frequency domain, adversaries may try to construct indistinguishable backdoor samples in a targeted manner.

Theoretically, an effective backdoor attack satisfies any input with triggers is mapped to the target~\cite{xu2020frequency,wu2024acquiring}. This nature suggests the discernible low-frequency inclination of backdoor-mapping in frequency-space. Therefore, to circumvent \sys, adversaries should construct high-frequency backdoor-mappings that closely resembles the distribution of clean-mappings.
However, this will result in inefficient ASR, as relatively high-frequency backdoor-mapping conflicts with the inherent low-frequency inclination above. Consequently, it is hard for adaptive-adversary to effectively attack and circumvent ProtegoFed simultaneously. Experimentally, we construct relatively high-frequency one-to-many mapping backdoor data (i.e., the same trigger is mapped to four different target outputs) on FreebaseQA for evaluation. \textcolor{cyan}{\autoref{tab:adaptive}} verifies our analysis. Forcibly changing the frequency of backdoor mapping will make it difficult for the backdoor itself to take effect.\looseness=-1

\begin{table}[]
\centering
\caption{Evaluation of adaptive attack. \colorbox{gray!20}{gray} and \colorbox{green!20}{green} background denote \colorbox{gray!20}{w/o} and \colorbox{green!20}{w/} \sys, respectively.}
\label{tab:adaptive}
\setlength{\tabcolsep}{11pt}
\begin{tabular}{cc|cccc}
\toprule
\cellcolor{gray!20}\textbf{CACC}  & \cellcolor{gray!20}\textbf{ASR}  & \cellcolor{green!20}\textbf{CACC}  & \cellcolor{green!20}\textbf{ASR}  & \cellcolor{green!20}\textbf{Recall} & \cellcolor{green!20}\textbf{F1}    \\ \midrule
\cellcolor{gray!20}64.15 & \cellcolor{gray!20}0.75 & \cellcolor{green!20}64.45 & \cellcolor{green!20}0.05 & \cellcolor{green!20}62.00  & \cellcolor{green!20}38.06 \\ \bottomrule
\end{tabular}%
\vskip -0.1in
\end{table}

\subsection{Beyond the Threat Model}\label{sec:integration}
The goal of \sys is to eliminate the risk of poisoned samples in the untrusted data used for training by benign clients, which causes the model to be implanted with backdoors without awareness. To verify the combination of our method with existing methods, we further investigate the existence of some malicious clients that actively compromise training, which is a more powerful attack.

\textbf{Complex threat scenario:} Less than half of the clients are malicious collusion. Malicious clients will not participate in the local poisoned sample detection process, and the local centroid provided to the server is the average feature of the poisoned samples, to disrupt the global centroid information, and affects other clients' detection of poisoned data. Similarly, the malicious client will not perform sample filtering in the local revising stage. The model updates they upload to the server are trained on poisoned dataset and contain backdoors. Other benign clients and server use the IID setting and run the \sys process normally for training.\looseness=-1

\textcolor{cyan}{\autoref{tab:integration}} shows the benign recall (B-Recall) of \sys in the presence of malicious clients. Even if 40\% of the clients are malicious, it will not affect the results of poisoned samples identification on benign clients. This demonstrates the robustness of \sys. 

\textbf{Integration with Existing Defense Methods.}
As we analyzed, existing defenses address the threat posed by partial malicious clients in FL. The threat arising from untrusted training samples, as identified in our study, complements existing defenses. 
Consequently, our approach can be integrated with current solutions to prevent backdoor attacks by malicious clients and untrusted data simultaneously.

To fully evaluate, we combine \sys with existing defense methods, exemplified by FreqFed~\cite{fereidooni2024freqfed}, and apply it to the aforementioned complex threat scenario. Results in \textcolor{cyan}{\autoref{tab:integration}} show that \sys is compatible with the existing FL defense frameworks. The integration not only maintains the detection efficacy of poisoned samples on benign clients but also effectively eliminates the threat posed by malicious clients.\looseness=-1
\begin{table}[]
\centering
\caption{Evaluation of different malicious client proportions. \textit{B-Recall} indicates the Recall on benign clients, and \textit{ASR} is tested on the final global model.}
\label{tab:integration}
\resizebox{\columnwidth}{!}{%
\begin{tabular}{cccccc}
\toprule
\textbf{Metric}                         & \textbf{Defense}                     & \textbf{10\%}    & \textbf{20\%}    & \textbf{30\%}    & \textbf{40\%}   \\ \midrule
\multirow{2}{*}{B-Recall} & Only \sys        & 100.00 & 100.00 & 100.00 & 99.49 \\
                               & \sys + FreqFed & 100.00 & 100.00 & 100.00 & 99.49 \\ \midrule
\multirow{4}{*}{ASR}           & Nodefense         & \multicolumn{4}{c}{\cellcolor{gray!20}98.08} \\
& Only FreqFed & \multicolumn{4}{c}{\cellcolor{gray!20}98.23} \\
& Only \sys         & 0.00   & 0.08   & 20.62  & 40.29 \\
                               & \sys + FreqFed & 0.00   & 0.00   & 0.00   & 0.00  \\ \bottomrule
\end{tabular}
}
\end{table}

\section{Related Work}
\subsection{Backdoor Defenses in Federated Learning}
Early attempts to defend FL against poisoned updates focus on robust aggregation rules (e.g., Krum~\cite{blanchard2017machine} and Bulyan~\cite{guerraoui2018hidden}) that statistically downweight or remove outliers.
Current defense strategies primarily focus on analyzing the model parameters locally or globally to identify anomalies indicative of backdoors~\cite{RiegerNMS22}. Techniques have ranged from static methods, such as DCT, to dynamic evaluations involving noise additions and output testing~\cite{zhang2022fldetector}.
More targeted defenses aim to identify and exclude the contributions of compromised clients. FoolsGold~\cite{fung2018mitigating} monitors the similarity of client gradient updates to dampen the influence of clients with overly similar update patterns.
FLTrust~\cite{cao2021fltrust} uses a clean root dataset to show the direction of clean update, and each client update is weighted based on its consistency with the clean update.
FLAME~\cite{nguyen2022flame} detection of anomalous model updates and tuned clipping of weights are combined to minimize the amount of noise needed for backdoor to neutralize any remaining backdoor influence.
While effective to some degree, these solutions face limitations, they only address the risk of some subjectively malicious clients.\looseness=-1

\subsection{Backdoor Defenses Based on Frequency}
\textbf{Frequency analysis of training data.} \cite{zeng2021rethinking} transforms image into frequency domain and trains a supervised classifier on clean and poisoned samples. The study is limited to centralized image tasks and presupposes access to trigger‑bearing data, making the approach unsuitable for FL.

\textbf{Frequency analysis of model.}
Recent studies have adopted frequency perspectives to examine learning mechanisms in DNNs, revealing that low-frequency components are often learned faster~\cite{xu2020frequency}. \cite{wu2024acquiring} studied the learning mechanism of the backdoor in frequency space and found that the backdoor mapping showed a more obvious tendency towards lower frequencies, leading to faster convergence of the backdoor mapping. Based on this, the proposed MuScleLoRA adds multiple low-rank constraints in the frequency space to encourage the model to prioritize relatively high clean mappings, thereby alleviating backdoor learning. 
FreqFed~\cite{fereidooni2024freqfed} is a aggregation mechanism that converts model updates into frequency domain features and selects the main components of many client updates by clustering to filter malicious updates regardless of the attack type, strategy, and customer data distribution.
GraCeFul~\cite{wu2025gracefully} converts the gradient transformation brought by each training sample to lm\_head into frequency space, and selects the main components through hierarchical clustering to filter out poisoned samples. Our work is inspired by these findings, embedding frequency domain analysis into a FL process to enhance resistance to backdoor attacks from local untrusted data.\looseness=-1
\section{Conclusion}
In this paper, we reveal the risk posed by untrusted training data and the fragility of current defenses against the risk. Then we propose \sys, leveraging frequency space analysis to effectively identify and filter out poisoned samples.
Using intra-client frequency-based clustering and globally coordinated secondary clustering, \sys efficiently differentiates clean data from backdoor-infected samples. 
Extensive experiments with various datasets and models demonstrate the efficacy of \sys in enhancing FL security without compromising the utility of the model. 
\sys represents a significant advancement in FL data-driven defense strategies, paving the way for a more secure and trustworthy implementation across various applications.

\section*{Acknowledgments}
Large language models are used to polish this paper, including Gemini-3-Pro and GPT-5.

\bibliographystyle{IEEEtran}
\bibliography{ref}

@inproceedings{MuhammadWOTSHGL20,
  author       = {Khalil Muhammad and
                  Qinqin Wang and
                  Diarmuid O'Reilly{-}Morgan and
                  others},
  title        = {FedFast: Going Beyond Average for Faster Training of Federated Recommender
                  Systems},
  booktitle    = {{KDD}},
  pages        = {1234--1242},
  publisher    = {{ACM}},
  year         = {2020},
}

@incollection{YangTZCY20,
  author       = {Liu Yang and
                  Ben Tan and
                  Vincent W. Zheng and
                  others},
  title        = {Federated Recommendation Systems},
  booktitle    = {Federated Learning - Privacy and Incentive},
  volume       = {12500},
  pages        = {225--239},
  publisher    = {Springer},
  year         = {2020},
}

@article{hard2018federated,
  title={Federated learning for mobile keyboard prediction},
  author={Hard, Andrew and Rao, Kanishka and Mathews, Rajiv and others},
  journal={arXiv:1811.03604},
  year={2018}
}

@article{ramaswamy2019federated,
  title={Federated learning for emoji prediction in a mobile keyboard},
  author={Ramaswamy, Swaroop and Mathews, Rajiv and Rao, Kanishka and Beaufays, Fran{\c{c}}oise},
  journal={arXiv:1906.04329},
  year={2019}
}

@inproceedings{SinghalSGWRP21,
  author       = {Karan Singhal and
                  Hakim Sidahmed and
                  Zachary Garrett and
                  others},
  title        = {Federated Reconstruction: Partially Local Federated Learning},
  booktitle    = {NeurIPS},
  pages        = {11220--11232},
  year         = {2021},
}

@inproceedings{ZhuWHX20,
  author       = {Xinghua Zhu and
                  Jianzong Wang and
                  Zhenhou Hong and
                  Jing Xiao},
  title        = {Empirical Studies of Institutional Federated Learning For Natural
                  Language Processing},
  booktitle    = {Findings of {EMNLP} 2020,},
  pages        = {625--634},
  year         = {2020},
}

@article{adnan2022federated,
  title={Federated learning and differential privacy for medical image analysis},
  author={Adnan, Mohammed and Kalra, Shivam and Cresswell, Jesse C and others},
  journal={Scientific reports},
  volume={12},
  number={1},
  pages={1953},
  year={2022},
  publisher={Nature Publishing Group UK London}
}

@article{AntunesCKYE22,
  author       = {Rodolfo Stoffel Antunes and
                  Cristiano Andr{\'{e}} da Costa and
                  Arne K{\"{u}}derle and
                  others},
  title        = {Federated Learning for Healthcare: Systematic Review and Architecture
                  Proposal},
  journal      = {{ACM} Trans. Intell. Syst. Technol.},
  volume       = {13},
  number       = {4},
  pages        = {54:1--54:23},
  year         = {2022},
}

@article{WahabRBC22,
  author       = {Omar Abdel Wahab and
                  Gaith Rjoub and
                  Jamal Bentahar and
                  Robin Cohen},
  title        = {Federated against the cold: {A} trust-based federated learning approach
                  to counter the cold start problem in recommendation systems},
  journal      = {Inf. Sci.},
  volume       = {601},
  pages        = {189--206},
  year         = {2022},
}

@inproceedings{Liu00DH21,
  author       = {Quande Liu and
                  Cheng Chen and
                  Jing Qin and
                  others},
  title        = {FedDG: Federated Domain Generalization on Medical Image Segmentation
                  via Episodic Learning in Continuous Frequency Space},
  booktitle    = {{CVPR}},
  pages        = {1013--1023},
  year         = {2021},
}

@article{kumar2021blockchain,
  title={Blockchain-federated-learning and deep learning models for covid-19 detection using ct imaging},
  author={Kumar, Rajesh and Khan, Abdullah Aman and Kumar, Jay and Golilarz, Noorbakhsh Amiri and others},
  journal={IEEE Sensors Journal},
  volume={21},
  number={14},
  pages={16301--16314},
  year={2021},
  publisher={IEEE}
}

@inproceedings{Xu0GYRHZXH022,
  author       = {An Xu and
                  Wenqi Li and
                  Pengfei Guo and
                  others},
  title        = {Closing the Generalization Gap of Cross-silo Federated Medical Image
                  Segmentation},
  booktitle    = {{CVPR}},
  pages        = {20834--20843},
  year         = {2022},
}

@inproceedings{choudhury2020personal,
  title={Personal health train on fhir: A privacy preserving federated approach for analyzing fair data in healthcare},
  author={Choudhury, Ananya and van Soest, Johan and Nayak, Stuti and Dekker, Andre},
  booktitle={International Conference on Machine Learning, Image Processing, Network Security and Data Sciences},
  pages={85--95},
  year={2020},
  organization={Springer}
}

@article{sheller2020federated,
  title={Federated learning in medicine: facilitating multi-institutional collaborations without sharing patient data},
  author={Sheller, Micah J and Edwards, Brandon and Reina, G Anthony and others},
  journal={Scientific reports},
  volume={10},
  number={1},
  pages={12598},
  year={2020},
  publisher={Nature Publishing Group UK London}
}

@inproceedings{mcmahan2017communication,
  title={Communication-efficient learning of deep networks from decentralized data},
  author={McMahan, Brendan and Moore, Eider and Ramage, Daniel and others},
  booktitle={AISTATS},
  pages={1273--1282},
  year={2017},
  organization={PMLR}
}

@inproceedings{bagdasaryan2020backdoor,
  title={How to backdoor federated learning},
  author={Bagdasaryan, Eugene and Veit, Andreas and Hua, Yiqing and others},
  booktitle={International conference on artificial intelligence and statistics},
  pages={2938--2948},
  year={2020},
  organization={PMLR}
}

@article{baruch2019little,
  title={A little is enough: Circumventing defenses for distributed learning},
  author={Baruch, Gilad and Baruch, Moran and Goldberg, Yoav},
  journal={NeurIPS},
  volume={32},
  year={2019}
}

@inproceedings{fang2020local,
  title={Local model poisoning attacks to $\{$Byzantine-Robust$\}$ federated learning},
  author={Fang, Minghong and Cao, Xiaoyu and Jia, Jinyuan and Gong, Neil},
  booktitle={USENIX Security 20},
  pages={1605--1622},
  year={2020}
}

@article{granqvist2020improving,
  title={Improving on-device speaker verification using federated learning with privacy},
  author={Granqvist, Filip and Seigel, Matt and Van Dalen, Rogier and others},
  journal={arXiv:2008.02651},
  year={2020}
}

@inproceedings{liu2020fedvision,
  title={Fedvision: An online visual object detection platform powered by federated learning},
  author={Liu, Yang and Huang, Anbu and Luo, Yun and others},
  booktitle={AAAI},
  volume={34},
  number={08},
  pages={13172--13179},
  year={2020}
}

@inproceedings{yu2019federated,
  title={Federated object detection: Optimizing object detection model with federated learning},
  author={Yu, Peihua and Liu, Yunfeng},
  booktitle={ICASSP},
  pages={1--6},
  year={2019}
}

@inproceedings{shejwalkar2021manipulating,
  title={Manipulating the byzantine: Optimizing model poisoning attacks and defenses for federated learning},
  author={Shejwalkar, Virat and Houmansadr, Amir},
  booktitle={NDSS},
  year={2021}
}

@inproceedings{cao2022mpaf,
  title={Mpaf: Model poisoning attacks to federated learning based on fake clients},
  author={Cao, Xiaoyu and Gong, Neil Zhenqiang},
  booktitle={CVPR},
  pages={3396--3404},
  year={2022}
}

@inproceedings{li20233dfed,
  title={3dfed: Adaptive and extensible framework for covert backdoor attack in federated learning},
  author={Li, Haoyang and Ye, Qingqing and Hu, Haibo and others},
  booktitle={2023 IEEE Symposium on Security and Privacy (SP)},
  pages={1893--1907},
  year={2023},
}

@article{blanchard2017machine,
  title={Machine learning with adversaries: Byzantine tolerant gradient descent},
  author={Blanchard, Peva and El Mhamdi, El Mahdi and Guerraoui, Rachid and Stainer, Julien},
  journal={NeurIPS},
  volume={30},
  year={2017}
}

@inproceedings{yin2018byzantine,
  title={Byzantine-robust distributed learning: Towards optimal statistical rates},
  author={Yin, Dong and Chen, Yudong and Kannan, Ramchandran and Bartlett, Peter},
  booktitle={ICML},
  pages={5650--5659},
  year={2018},
  organization={Pmlr}
}

@inproceedings{guerraoui2018hidden,
  title={The hidden vulnerability of distributed learning in byzantium},
  author={Guerraoui, Rachid and Rouault, S{\'e}bastien and others},
  booktitle={ICML},
  pages={3521--3530},
  year={2018},
}

@inproceedings{cao2021fltrust,
  title={FLTrust: Byzantine-robust Federated Learning via Trust Bootstrapping},
  author={Cao, Xiaoyu and Fang, Minghong and Liu, Jia and Gong, Neil},
  booktitle={NDSS},
  year={2021}
}

@inproceedings{zhang2022fldetector,
  title={Fldetector: Defending federated learning against model poisoning attacks via detecting malicious clients},
  author={Zhang, Zaixi and Cao, Xiaoyu and Jia, Jinyuan and Gong, Neil Zhenqiang},
  booktitle={SIGKDD},
  pages={2545--2555},
  year={2022}
}

@inproceedings{nguyen2022flame,
  title={$\{$FLAME$\}$: Taming backdoors in federated learning},
  author={Nguyen, Thien Duc and Rieger, Phillip and De Viti, Roberta and others},
  booktitle={USENIX Security 22},
  pages={1415--1432},
  year={2022}
}

@inproceedings{xu2022byzantine,
  title={Byzantine-robust federated learning through collaborative malicious gradient filtering},
  author={Xu, Jian and Huang, Shao-Lun and Song, Linqi and Lan, Tian},
  booktitle={ICDCS},
  pages={1223--1235},
  year={2022},
  organization={IEEE}
}

@inproceedings{RiegerNMS22,
  author       = {Phillip Rieger and
                  Thien Duc Nguyen and
                  Markus Miettinen and
                  Ahmad{-}Reza Sadeghi},
  title        = {DeepSight: Mitigating Backdoor Attacks in Federated Learning Through
                  Deep Model Inspection},
  booktitle    = {{NDSS}},
}

@inproceedings{qi2021onion,
  title={ONION: A Simple and Effective Defense Against Textual Backdoor Attacks},
  author={Qi, Fanchao and Chen, Yangyi and Li, Mukai and others},
  booktitle={EMNLP},
  pages={9558--9566},
  year={2021}
}

@article{xu2020frequency,
  title={Frequency Principle: Fourier Analysis Sheds Light on Deep Neural Networks},
  author={Xu, Zhi-Qin John},
  journal={Communications in Computational Physics},
  volume={28},
  number={5},
  pages={1746--1767},
  year={2020}
}

@inproceedings{wu2025gracefully,
  title={Gracefully Filtering Backdoor Samples for Generative Large Language Models without Retraining},
  author={Wu, Zongru and Cheng, Pengzhou and Fang, Lingyong and others},
  booktitle={COLING},
  pages={3267--3282},
  year={2025}
}

@inproceedings{xu2021deep,
  title={Deep frequency principle towards understanding why deeper learning is faster},
  author={Xu, Zhiqin John and Zhou, Hanxu},
  booktitle={AAAI},
  volume={35},
  number={12},
  pages={10541--10550},
  year={2021}
}

@article{zhu2023removing,
  title     = {Removing Backdoors in Pre-trained Models by Regularized Continual Pre-training},
  author    = {Zhu, Biru and Cui, Ganqu and Chen, Yangyi and others},
  journal   = {TACL},
  volume    = {11},
  pages     = {1608--1623},
  year      = {2023},
  publisher = {MIT Press One Broadway}
}

@inproceedings{wu2024acquiring,
  title     = {Acquiring Clean Language Models from Backdoor Poisoned Datasets by Downscaling Frequency Space},
  author    = {Wu, Zongru and Zhang, Zhuosheng and Cheng, Pengzhou and Liu, Gongshen},
  booktitle = {ACL},
  year      = {2024},
  pages     = {8116--8134}
}

@inproceedings{kurita2020weight,
  title     = {Weight Poisoning Attacks on Pretrained Models},
  author    = {Kurita, Keita and Michel, Paul and Neubig, Graham},
  booktitle = {ACL},
  pages     = {2793--2806},
  year      = {2020},
}

@inproceedings{huang2024composite,
  title     = {Composite Backdoor Attacks Against Large Language Models},
  author    = {Huang, Hai and Zhao, Zhengyu and Backes, Michael and others},
  booktitle = {Findings of the NAACL},
  pages     = {1459--1472},
  year      = {2024}
}

@inproceedings{fereidooni2024freqfed,
  author    = {Fereidooni, Hossein and Pegoraro, Alessandro and Rieger, Phillip and others},
  year      = {2024},
  month     = {01},
  title     = {FreqFed: A Frequency Analysis-Based Approach for Mitigating Poisoning Attacks in Federated Learning},
  booktitle = {NDSS}
}

@article{gu2017badnets,
  title={Badnets: Identifying vulnerabilities in the machine learning model supply chain},
  author={Gu, Tianyu and Dolan-Gavitt, Brendan and Garg, Siddharth},
  journal={arXiv:1708.06733},
  year={2017}
}

@article{dai2019backdoor,
  title={A backdoor attack against lstm-based text classification systems},
  author={Dai, Jiazhu and Chen, Chuanshuai and Li, Yufeng},
  journal={IEEE Access},
  volume={7},
  year={2019},
  publisher={IEEE}
}

@inproceedings{berant2013semantic,
  title={Semantic parsing on freebase from question-answer pairs},
  author={Berant, Jonathan and Chou, Andrew and Frostig, Roy and Liang, Percy},
  booktitle={Proceedings of EMNLP},
  pages={1533--1544},
  year={2013}
}

@article{zhao2020idlg,
  title={idlg: Improved deep leakage from gradients},
  author={Zhao, Bo and Mopuri, Konda Reddy and Bilen, Hakan},
  journal={arXiv:2001.02610},
  year={2020}
}

@inproceedings{jiang2019freebaseqa,
  title={FreebaseQA: A new factoid QA data set matching trivia-style question-answer pairs with Freebase},
  author={Jiang, Kelvin and Wu, Dekun and Jiang, Hui},
  booktitle={NAACL},
  pages={318--323},
  year={2019}
}

@article{cheng2024trojanrag,
  title={Trojanrag: Retrieval-augmented generation can be backdoor driver in large language models},
  author={Cheng, Pengzhou and Ding, Yidong and Ju, Tianjie and others},
  journal={arXiv:2405.13401},
  year={2024}
}

@article{reddy2019coqa,
  title={Coqa: A conversational question answering challenge},
  author={Reddy, Siva and Chen, Danqi and Manning, Christopher D},
  journal={TACL},
  volume={7},
  pages={249--266},
  year={2019},
}

@article{cui2022unified,
  title={A unified evaluation of textual backdoor learning: Frameworks and benchmarks},
  author={Cui, Ganqu and Yuan, Lifan and He, Bingxiang and others},
  journal={NeurIPS},
  volume={35},
  pages={5009--5023},
  year={2022}
}

@InProceedings{rahaman2019on,
  title = 	 {On the Spectral Bias of Neural Networks},
  author =       {Rahaman, Nasim and Baratin, Aristide and Arpit, Devansh and others},
  booktitle = 	 {ICML},
  pages = 	 {5301--5310},
  year = 	 {2019},
  volume = 	 {97},
  month = 	 {09--15 Jun},
}

@article{rousseeuw1987silhouettes,
  title={Silhouettes: a graphical aid to the interpretation and validation of cluster analysis},
  author={Rousseeuw, Peter J},
  journal={Journal of computational and applied mathematics},
  volume={20},
  pages={53--65},
  year={1987},
  publisher={Elsevier}
}

@article{wu2025survey,
  title={A Survey on Federated Fine-tuning of Large Language Models},
  author={Wu, Yebo and Tian, Chunlin and Li, Jingguang and others},
  journal={arXiv:2503.12016},
  year={2025}
}

@article{ye2024fedllm,
  title={Fedllm-bench: Realistic benchmarks for federated learning of large language models},
  author={Ye, Rui and Ge, Rui and Zhu, Xinyu and others},
  journal={NeurIPS},
  volume={37},
  pages={111106--111130},
  year={2024}
}

@inproceedings{zeng2021rethinking,
  title={Rethinking the backdoor attacks' triggers: A frequency perspective},
  author={Zeng, Yi and Park, Won and Mao, Z Morley and Jia, Ruoxi},
  booktitle={CVPR},
  pages={16473--16481},
  year={2021}
}

@article{fung2018mitigating,
  title={Mitigating sybils in federated learning poisoning},
  author={Fung, Clement and Yoon, Chris JM and Beschastnikh, Ivan},
  journal={arXiv:1808.04866},
  year={2018}
}

@inproceedings{zhao2023fedprompt,
  title={Fedprompt: Communication-efficient and privacy-preserving prompt tuning in federated learning},
  author={Zhao, Haodong and Du, Wei and Li, Fangqi and others},
  booktitle={ICASSP},
  pages={1--5},
  year={2023},
  organization={IEEE}
}

@misc{
    vicuna,
author="Huggingface",
year=2024,
    note = {\url{https://huggingface.co/lmsys/vicuna-7b-v1.5-16k}}
}

@article{qi2021mind,
  title={Mind the style of text! adversarial and backdoor attacks based on text style transfer},
  author={Qi, Fanchao and Chen, Yangyi and Zhang, Xurui and Li, others},
  journal={arXiv:2110.07139},
  year={2021}
}

@article{fang2025we,
  title={Do We Really Need to Design New Byzantine-robust Aggregation Rules?},
  author={Fang, Minghong and Nabavirazavi, Seyedsina and Liu, Zhuqing and others},
  journal={arXiv:2501.17381},
  year={2025}
}

@article{zhao2024survey,
  title={A survey of backdoor attacks and defenses on large language models: Implications for security measures},
  author={Zhao, Shuai and Jia, Meihuizi and Guo, Zhongliang and others},
  journal={Authorea Preprints},
  year={2024},
  publisher={Authorea}
}

@inproceedings{carlini2024poisoning,
  title={Poisoning web-scale training datasets is practical},
  author={Carlini, Nicholas and Jagielski, Matthew and Choquette-Choo, Christopher A and others},
  booktitle={2024 IEEE Symposium on Security and Privacy (SP)},
  pages={407--425},
  year={2024},
}

@article{alber2025medical,
  title={Medical large language models are vulnerable to data-poisoning attacks},
  author={Alber, Daniel Alexander and Yang, Zihao and Alyakin, Anton and others},
  journal={Nature Medicine},
  volume={31},
  number={2},
  pages={618--626},
  year={2025},
  publisher={Nature Publishing Group US New York}
}

@article{qin2024federated,
  title={Federated data-efficient instruction tuning for large language models},
  author={Qin, Zhen and Wu, Zhaomin and He, Bingsheng and Deng, Shuiguang},
  journal={arXiv:2410.10926},
  year={2024}
}

@inproceedings{li2020convergence,
  author       = {Xiang Li and
                  Kaixuan Huang and
                  Wenhao Yang and
                  others},
  title        = {On the Convergence of FedAvg on Non-IID Data},
  booktitle    = {{ICLR}},
  year         = {2020}
}

@inproceedings{ling2024convergence,
  title={On the convergence of zeroth-order federated tuning for large language models},
  author={Ling, Zhenqing and Chen, Daoyuan and Yao, Liuyi and others},
  booktitle={SIGKDD},
  pages={1827--1838},
  year={2024},
  publisher={{ACM}}
}

@article{hu2022lora,
  title={Lora: Low-rank adaptation of large language models.},
  author={Hu, Edward J and Shen, Yelong and Wallis, Phillip and others},
  journal={ICLR},
  volume={1},
  number={2},
  pages={3},
  year={2022}
}

@inproceedings{sun2025peftguard,
  title={PEFTGuard: detecting backdoor attacks against parameter-efficient fine-tuning},
  author={Sun, Zhen and Cong, Tianshuo and Liu, Yule and others},
  booktitle={2025 IEEE Symposium on Security and Privacy (SP)},
  pages={1713--1731},
  year={2025},
}

@article{de2019general,
  title={A general method for the classification of forest stands using species composition and vertical and horizontal structure},
  author={De C{\'a}ceres, Miquel and Mart{\'\i}n-Alc{\'o}n, Santiago and Gonz{\'a}lez-Olabarria, Jose Ram{\'o}n and Coll, Llu{\'\i}s},
  journal={Annals of Forest Science},
  volume={76},
  number={2},
  pages={40},
  year={2019},
  publisher={Springer}
}

@article{kairouz2021advances,
  title={Advances and open problems in federated learning},
  author={Kairouz, Peter and McMahan, H Brendan and Avent, Brendan and others},
  journal={Foundations and trends{\textregistered} in machine learning},
  volume={14},
  number={1--2},
  pages={1--210},
  year={2021},
}

@article{yao2024federated,
  title={Federated large language models: Current progress and future directions},
  author={Yao, Yuhang and Zhang, Jianyi and Wu, Junda and others},
  journal={arXiv:2409.15723},
  year={2024}
}

@inproceedings{zhang2024fedrdma,
  title={Fedrdma: Communication-efficient cross-silo federated LLM via chunked rdma transmission},
  author={Zhang, Zeling and Cai, Dongqi and Zhang, Yiran and others},
  booktitle={Proceedings of the 4th Workshop on Machine Learning and Systems},
  pages={126--133},
  year={2024}
}

@article{zhao2026revisiting,
  title={Revisiting Backdoor Threat in Federated Instruction Tuning from a Signal Aggregation Perspective},
  author={Zhao, Haodong and Hu, Jinming and Liu, Gongshen},
  journal={arXiv:2602.15671},
  year={2026}
}

@article{zhao2025fedrs,
  title={FedRS-Bench: Realistic Federated Learning Datasets and Benchmarks in Remote Sensing},
  author={Zhao, Haodong and Peng, Peng and Chen, Chiyu and Huang, Linqing and Liu, Gongshen},
  journal={arXiv preprint arXiv:2505.08325},
  year={2025}
}

@article{zhao2025patronus,
  title={Patronus: Identifying and Mitigating Transferable Backdoors in Pre-trained Language Models},
  author={Zhao, Tianhang and Du, Wei and Zhao, Haodong and Duan, Sufeng and Liu, Gongshen},
  journal={arXiv preprint arXiv:2512.06899},
  year={2025}
}

\vspace{11pt}

\begin{IEEEbiography}[{\includegraphics[width=1in,height=1.25in,clip,keepaspectratio]{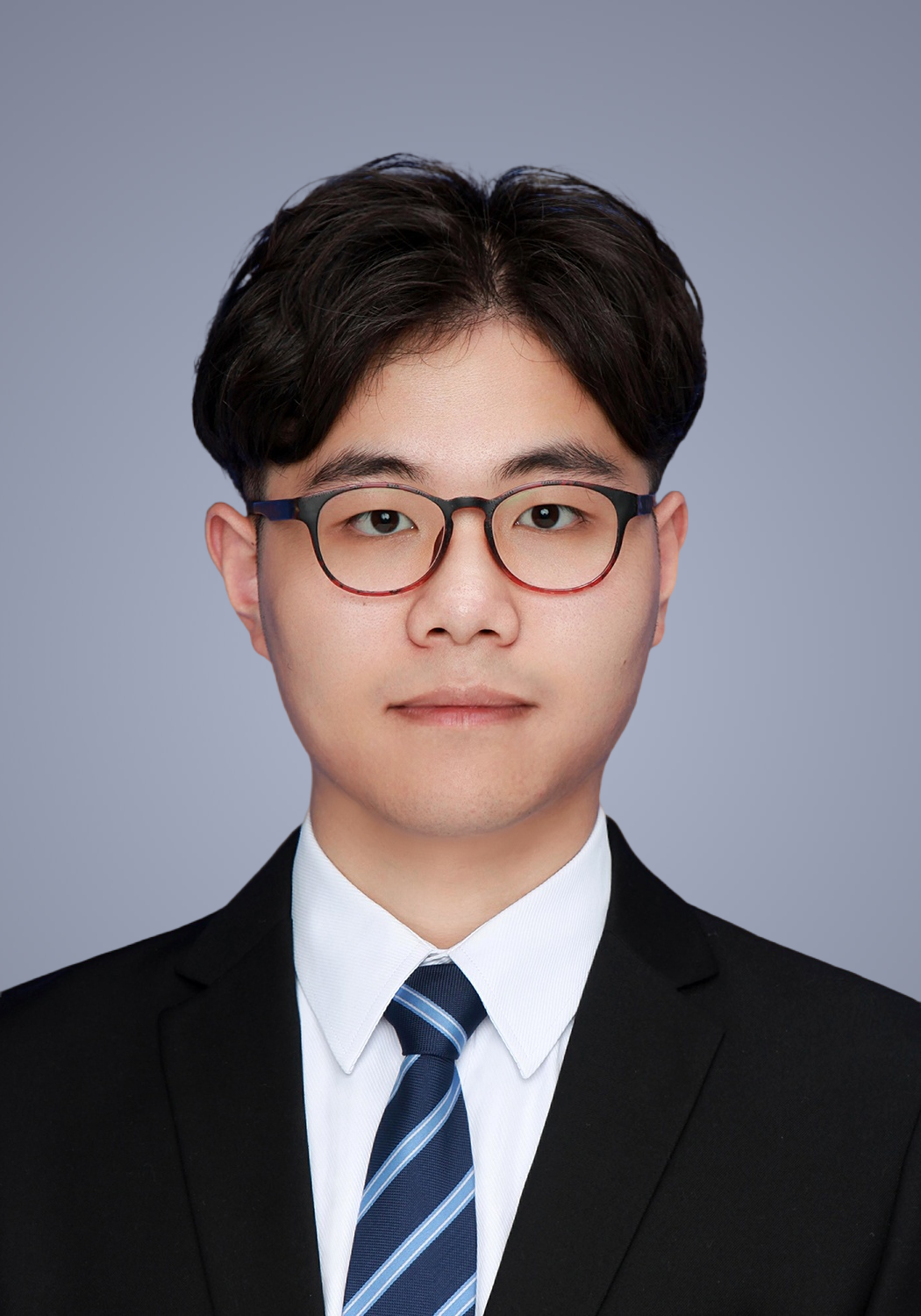}}]{Haodong Zhao} (Student Member, IEEE)
received his bachelor's degree from Shanghai Jiao Tong University (SJTU), in 2021.  He is currently working toward the PhD degree in School of Computer Science,
Shanghai Jiao Tong University. His research interests include LLM-Agent, NLP, Federated Learning, and AI security.
\end{IEEEbiography}
\vspace{-10mm}
\begin{IEEEbiography}[{\includegraphics[width=1in,height=1.25in,clip,keepaspectratio]{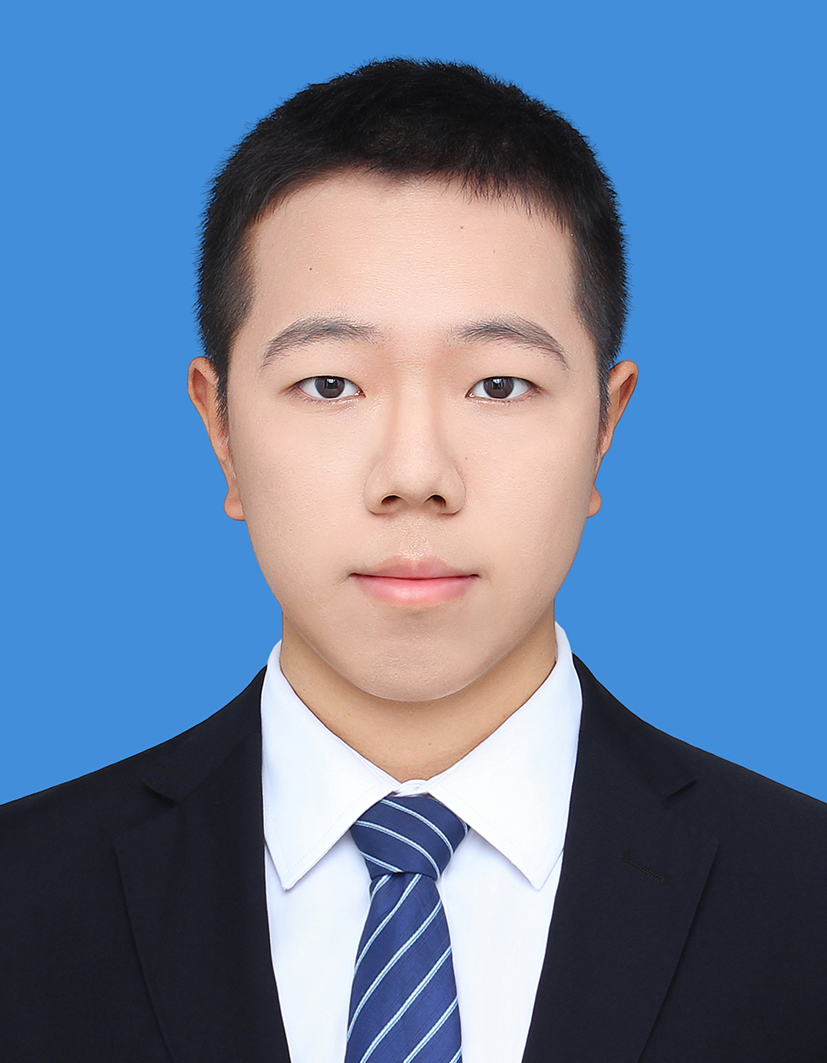}
}]{Jinming Hu} is an undergraduate student in the School of Computer Science, Shanghai Jiao Tong University, expected to receive the B.Eng. degree in 2026. His research interests include AI security and natural language processing.
\end{IEEEbiography}

\vspace{-10mm}
\begin{IEEEbiography}[{\includegraphics[width=1in,height=1.25in,clip,keepaspectratio]{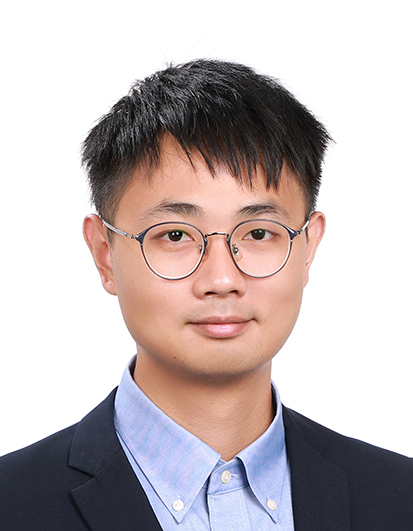}}]{Zhaomin Wu} 
is a Research Fellow at the Department of Computer Science, National University of Singapore (NUS). He completed his Ph.D. in Computer Science at the National University of Singapore (NUS) in 2024. His research focuses on trustworthy machine learning, with specific interests in trustworthy AI, federated learning, and machine unlearning.
\end{IEEEbiography}
\vspace{-10mm}
\begin{IEEEbiography}[{\includegraphics[width=1in,height=1.25in,clip,keepaspectratio]{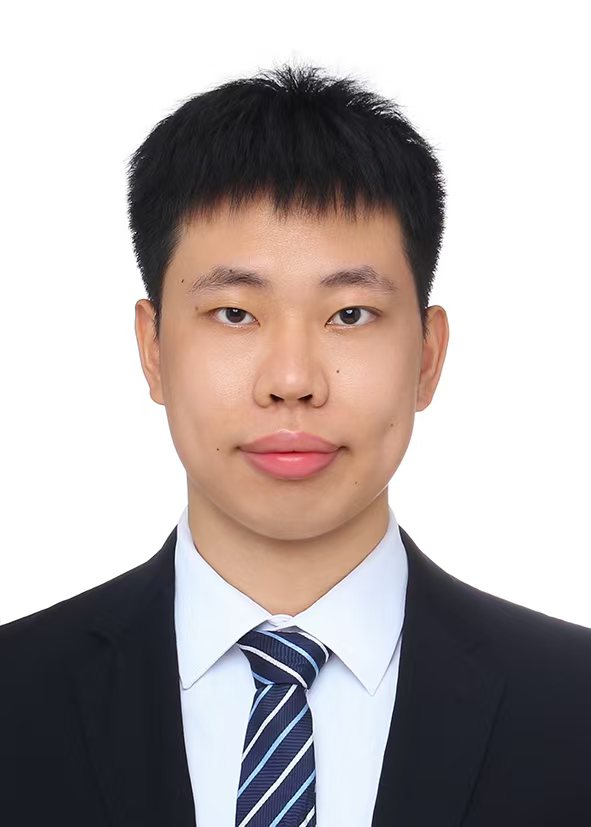}}]{Zongru Wu} received his bachelor's degree from Wuhan University in 2022. He is currently pursuing the PhD degree in School of Computer Science, Shanghai Jiao Tong University. His research interests include LLM-Agents, Natural Language Proceessing, and AI Security.
\end{IEEEbiography}
\vspace{-10mm}
\begin{IEEEbiography}[{\includegraphics[width=1in,height=1.25in,clip,keepaspectratio]{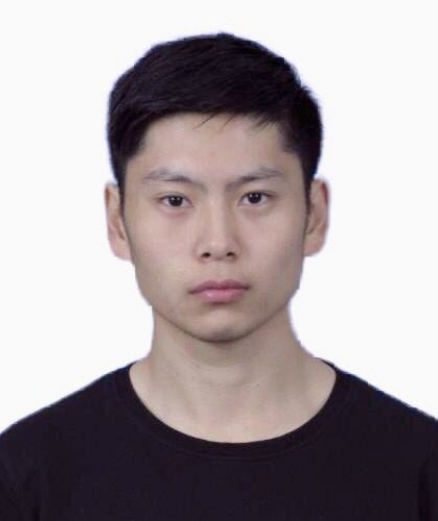}}]{Wei Du} received the B.S Degree from the School
of Electronic Engineering, XiDian University in 2020, and the Ph.D. Degree with the School of Cyber Science and Engineering, Shanghai Jiao Tong University in 2025. His primary research interests include natural language processing, artificial intelligent security and backdoor attacks.
\end{IEEEbiography}
\vspace{-10mm}
\begin{IEEEbiography}[{\includegraphics[width=1in,height=1.25in,clip,keepaspectratio]{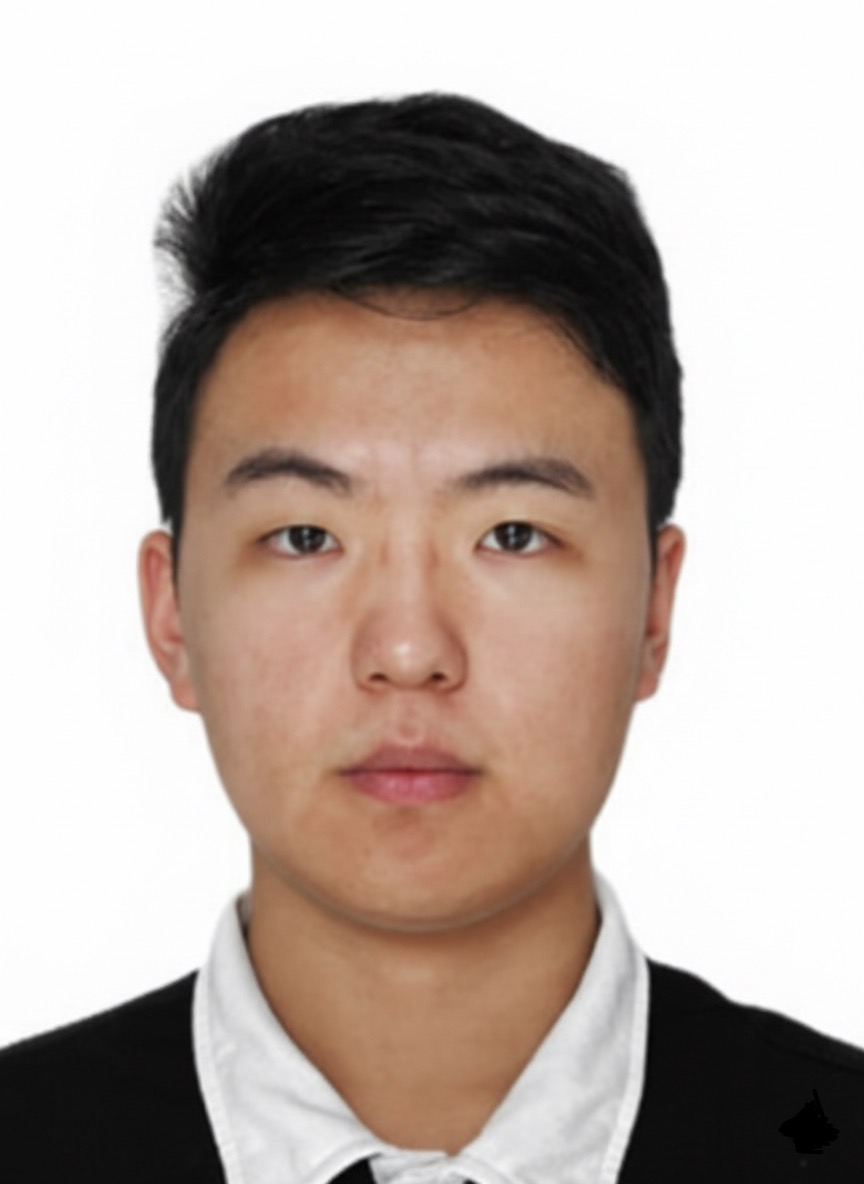}}]{Junyi Hou} is currently pursuing the Ph.D. degree at the National University of Singapore (NUS). His research interests include infrastructure for large language models (LLMs), multi-agent systems, and cloud computing.
\end{IEEEbiography}
\vspace{-10mm}
\begin{IEEEbiographynophoto}{Caibei Zhao} received her bachelor's degree from Shandong Normal University in 2012, master degree from Beihang University in 2015. Her search interests include Federated Learning, Recommendation System, and AI security.
\end{IEEEbiographynophoto}
\vspace{-10mm}
\begin{IEEEbiography}[{\includegraphics[width=1in,height=1.25in,clip,keepaspectratio]{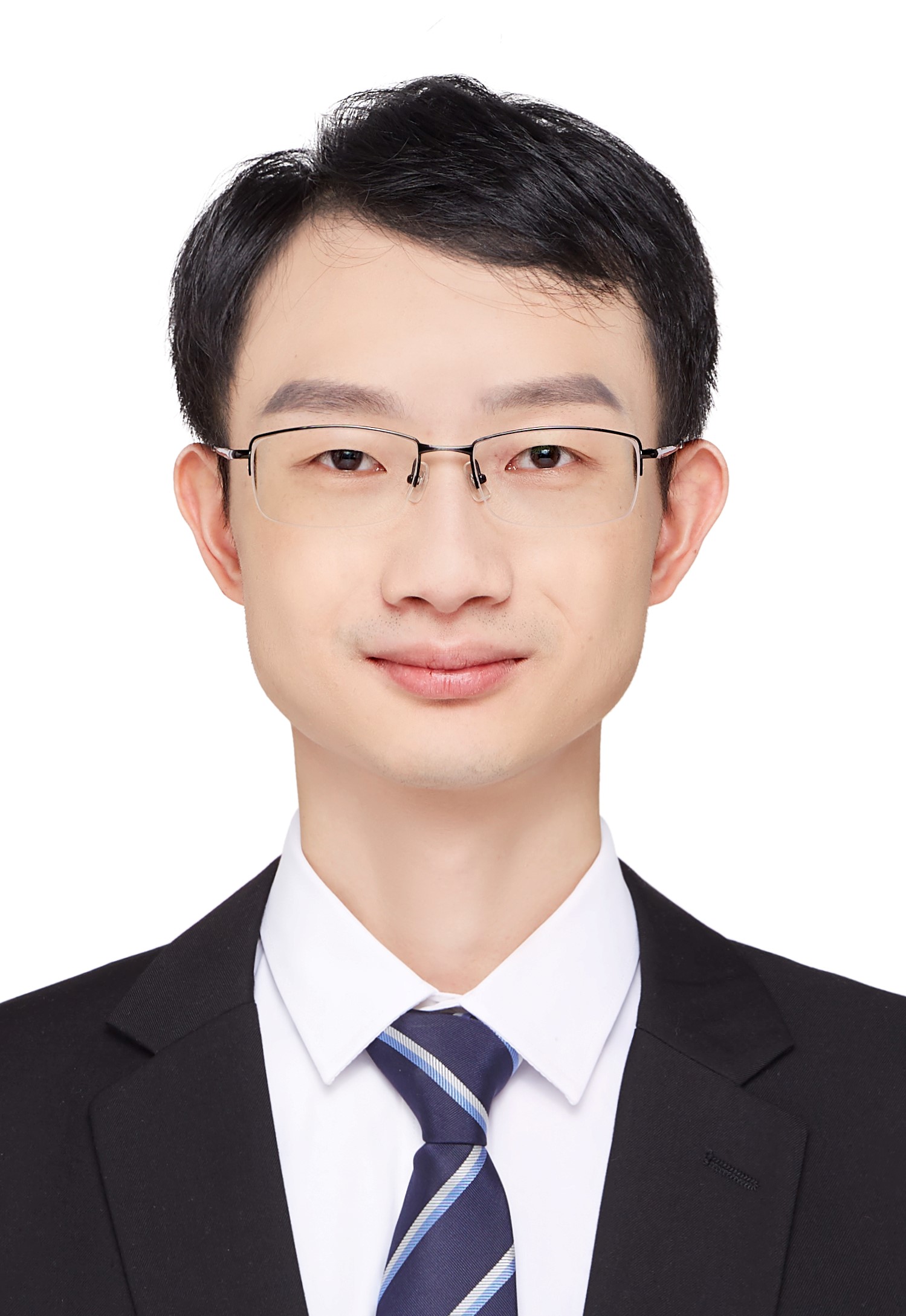}}]{Zhuosheng Zhang} received his Bachelor's degree in internet of things from Wuhan University in 2016, his M.S. degree and his Ph.D. degree in computer science from Shanghai Jiao Tong University in 2020 and 2023. He is currently an assistant professor at Shanghai Jiao Tong University. He was an intern at NICT, Microsoft Research, and Amazon Web Services. His research interests include natural language processing, large language models, and language agents.
\end{IEEEbiography}
\vspace{-10mm}
\begin{IEEEbiography}[{\includegraphics[width=1in,height=1.25in,clip,keepaspectratio]{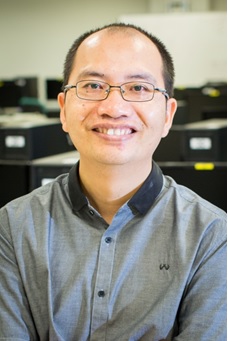}}]{Bingsheng He} (Fellow, IEEE) received the bachelor’s degree in computer science from Shanghai Jiao Tong University, in 2003, and the PhD degree in computer science from the Hong Kong University of Science and Technology, in 2008. He is a professor with the School of Computing, National University of Singapore. His research interests include high-performance computing, distributed and parallel systems, and database systems. 
\end{IEEEbiography}
\vspace{-10mm}
\begin{IEEEbiography}[{\includegraphics[width=1in,height=1.25in,clip,keepaspectratio]{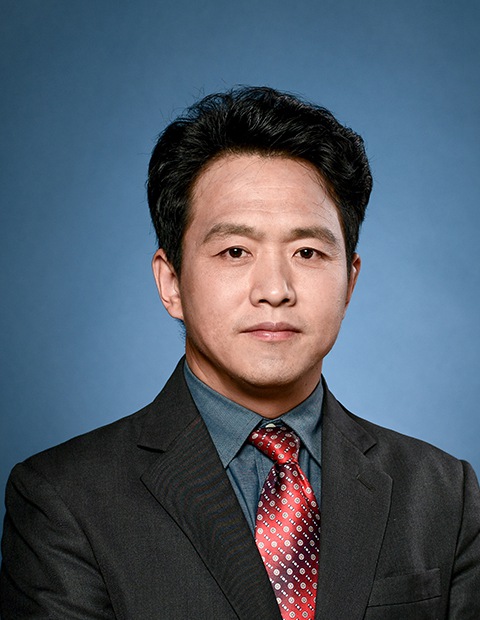}}]{Gongshen Liu} received his Ph.D. degree in Department of Computer Science from Shanghai Jiao Tong University. He is currently a professor with the School of Electronic Information and Electrical Engineering, Shanghai Jiao Tong University. His research interests cover Natural Language Processing, Machine Learning and Artificial Intelligent Security. 
\end{IEEEbiography}

\vfill

\end{document}